\documentclass[a4paper,twopages]{article}

\usepackage[latin1]{inputenc}
\usepackage[american]{babel}
\usepackage{amsmath,amsfonts,amssymb,amsthm}
\usepackage{paralist}
\usepackage{url,xspace}
\usepackage[ruled,vlined]{algorithm2e}
\usepackage{subfig}
\usepackage{graphicx,color,psfrag}
\graphicspath{{fig/},{logos/}}
\usepackage{pstool}
\usepackage{multirow}
\usepackage{rotating}

\newtheorem{lemma}{Lemma}
\newtheorem{theorem}{Theorem}
\newtheorem{proposition}{Proposition}

\newcommand{\elim}{\mathit{elim}}

\newcommand{\GEQRF}{\ensuremath{\mathit{GEQRF}}\xspace}
\newcommand{\GEQRTWO}{\ensuremath{\mathit{GEQR2}}\xspace}

\newcommand{\GEQRT}{\ensuremath{\mathit{GEQRT}}\xspace}
\newcommand{\TSQRT}{\ensuremath{\mathit{TSQRT}}\xspace}

\newcommand{\UNMQR}{\ensuremath{\mathit{UNMQR}}\xspace}
\newcommand{\TSMQR}{\ensuremath{\mathit{TSMQR}}\xspace}
\newcommand{\TTQRT}{\ensuremath{\mathit{TTQRT}}\xspace}
\newcommand{\TTMQR}{\ensuremath{\mathit{TTMQR}}\xspace}
\newcommand{\coarse}{\ensuremath{\mathit{coarse}}\xspace}
\newcommand{\BS}{\ensuremath{\mathit{BS}}\xspace}
\newcommand{\s}{\ensuremath{\star}\xspace}
\newcommand{\TSFT}{\textsc{TS-FlatTree}\xspace}
\newcommand{\TTFT}{\textsc{TT-FlatTree}\xspace}
\newcommand{\SK}{\textsc{Sameh-Kuck}\xspace}
\newcommand{\MC}{\textsc{Fibonacci}\xspace}
\newcommand{\BT}{\textsc{BinaryTree}\xspace}
\newcommand{\Greedy}{\textsc{Greedy}\xspace}
\newcommand{\FT}{\textsc{FlatTree}\xspace}
\newcommand{\PT}{\textsc{PlasmaTree}\xspace}
\newcommand{\ASAP}{\textsc{Asap}\xspace}
\newcommand{\GRASAP}{\textsc{Grasap}}
\newcommand{\Alg}{\textsc{Alg}\xspace}
\newcommand{\TiledQRURL}{\url{http://graal.ens-lyon.fr/~mjacquel/tiledQR.html}\xspace}

\newboolean{withproof}
\setboolean{withproof}{false}
\newboolean{withoutproof}
\setboolean{withoutproof}{true}
\newcommand{\forRR}[1]{\ifthenelse{\boolean{withproof}}{{#1}}{}}
\newcommand{\forSCpaper}[1]{\ifthenelse{\boolean{withoutproof}}{{#1}}{}}

\date{Avril 2011}

\author{
  Henricus Bouwmeester
\thanks{University of Colorado Denver}
\thanks{Research of the first author was fully supported by the National Science Foundation grant \# NSF CCF 811520.}
  \and
  Mathias Jacquelin
\thanks{\'Ecole Normale Sup\'erieure de Lyon}
  \and
  Julien Langou
\footnotemark[1]
\thanks{Research of the third author was fully supported by the National Science Foundation grant \# NSF CCF 1054864.}
  \and
  Yves Robert
\footnotemark[3]
\thanks{Institut Universitaire de France}
\thanks{Research of the fourth author was supported in part by the ANR StochaGrid and RESCUE projects.}
}

\title{Tiled QR factorization algorithms}

\begin{document}

\maketitle
\begin{abstract}
        This work revisits existing algorithms for the QR factorization of
    rectangular matrices composed of $p \times q$ tiles, where $p \geq q$.
    Within this framework, we study the critical paths and performance of
    algorithms such as \SK, \MC, \Greedy, and those found within PLASMA.
    Although neither \MC nor \Greedy is optimal, both are shown to
    be asymptotically optimal for all matrices of size $p = q^2 f(q)$, where
    $f$ is any function such that $\lim_{+\infty} f= 0$. This novel and important complexity
    result applies to
    all matrices where $p$ and $q$ are proportional, $ p = \lambda q$, with
    $\lambda \geq 1$, thereby encompassing many important situations in
    practice (least squares). We provide an extensive set of experiments that
    show the superiority of the new algorithms for
    tall matrices.
\end{abstract}

\section{Introduction}
\label{sec.intro}

Given an $m$-by-$n$ matrix $A$ with $n \leq m$, we consider the computation of
its QR factorization, which is the factorization $A = QR$, where $Q$ is an
$m$-by-$n$ unitary matrix ($Q^HQ = I_n$), and $R$ is upper triangular.

The QR factorization is the time consuming stage of some important numerical
computations.  The QR factorization of an $m$-by-$n$ matrix with $n \leq m$ is
needed for solving a linear least squares 
with $m$ equations (observations) and $n$ unknowns.  The QR factorization of an $m$-by-$n$ matrix
with $n \leq m$ is used to compute an orthogonal basis (the $Q$-factor) of the
column span of the initial matrix $A$.  For example, all block iterative
methods (used to solve large sparse linear systems of equations or computing
some relevant eigenvalues of such systems) require orthogonalizing a set of
vectors at each step of the process. For these two usage examples, while
$n\leq m$, $n$ can range from $n \ll m$ to $ n \lessapprox m$. We note that the
extreme case $n=m$ is also relevant:  the QR factorization of a matrix can be
used to solve (square) linear systems of equations (in that case $n=m$). While
this requires twice as many flops as an LU factorization, using a QR
factorization (a) is unconditionally stable (Gaussian elimination with partial
pivoting or pairwise pivoting is not) and (b) avoids pivoting so it may well be
faster in some cases (despite requiring twice as many flops).

To obtain a QR factorization, we consider algorithms which apply a
sequence of $m$-by-$m$ unitary transformations, $U_i$, ($U_i^HU_i=I$,),
$i=1,\dots,\ell$, on the left of the matrix $A$, such that after $\ell$
transformations the resulting matrix $ R = U_\ell \ldots U_1 A $ is upper
triangular, in which case, $R$ is indeed the $R$-factor of the QR factorization.
The $Q$-factor (if needed) can then be obtained by computing $ Q = U_1^H \ldots
U_\ell^H $.  These types of algorithms are in regular use, e.g. in the
LAPACK and ScaLAPACK libraries, and are favored over others algorithms (Cholesky
QR or Gram-Schmidt) for their stability.

The unitary transformation $U_i$ is chosen so as to introduce some zeros in the
current update matrix $ U_{i-1} \ldots U_1 A $. The two basic transformations
are Givens rotations and Householder reflections.  One Givens rotation
introduces one additional zero; the whole triangularization requires $mn -
n(n+1)/2$ Givens rotations for $n<m$. One elementary Householder reflection
simultaneously introduces $m-i$ zeros in position $i+1$ to $m$ in column $i$;
the whole triangularization requires $n$ Householder reflections for $n<m$.
(See LAPACK subroutine \GEQRTWO.) The LAPACK \GEQRT subroutine constructs a
compact WY representation to apply a sequence of $i_b$ Householder reflections,
this enables one to introduce the appropriate zeros in $i_b$ consecutive columns
and thus leverage optimized Level 3 BLAS subroutines during the update. The
blocking of Givens rotations is also possible but is more costly in terms of
flops.

The main interest of Givens rotations over Householder transformations is that
one can concurrently introduce zeros using disjoint pairs of rows, in other
words, two transformations $U_i$ and $U_{i+1}$ may be applicable concurrently.
This is not possible using the original Householder reflection algorithm since
the transformations work on whole columns and thus does not exhibit this type
of intrinsic parallelism forcing this kind of Householder reflections to be
applied sequentially. The advantage of Householder reflections over Givens
rotations is that, first, Householder reflections perform less flops, and second,
the compact WY transformation enables high sequential performance of the
algorithm. In a multicore setting, where data locality and parallelism are
crucial algorithmic characteristics for enabling performance, the tiled QR
factorization algorithm combines both ideas: use of Householder reflections for
high sequential performance and use of a scheme ala Givens rotations to enable
parallelism within cores. In essence, one can think (i) either of the tiled QR
factorization as a Givens rotation scheme but on tiles ($m_b$-by-$n_b$
submatrices) instead of on scalars ($1$-by-$1$ submatrices) as in the original
scheme, (ii) or of it as a blocked Householder reflection scheme
where each reflection is confined to an extent much less than the full column
span, which enables concurrency with other reflections.

Tiled QR factorization in the context of multicore architectures has been
introduced in~\cite{Buttari2008,tileplasma,Quintana:2009}. Initially the focus
was on square matrices and the sequence of unitary transformations presented
was analogous to \SK~\cite{SamehKuck78}, which corresponds to reducing the panels with flat trees.
The possibility of using any tree in order to either
maximize parallelism or minimize communication is explained in~\cite{CAQR}.

The focus of this manuscript is in maximizing parallelism.  Stemming from the
observation that a binary tree is best for tall and skinny matrices and a flat
tree is best for square matrices, Hadri et al.~\cite{Hadri_ipdps_2010}, propose
to use trees which combine flat trees at the bottom level with a binary tree at
the top level in order to exhibit more parallelism.  Our theoretical and
experimental work explains that we can adapt \MC~\cite{ModiClarke84} and
\Greedy~\cite{j12,j14} to tiles, resulting in yet better algorithms in terms of
parallelism. Moreover our new algorithms do not have any tuning parameter such
as the domain size in the case of~\cite{Hadri_ipdps_2010}.

The focus of this manuscript is not in trying to reduce communication (data
movement between memory hierarchy) to a minimum.  Relatively low level of communication is
naturally achieved by the algorithm by tiling the operations.
How to optimize the trade-off communication and parallelism is out of the scope of this manuscript.
For this reason, we consider square tiling with constant tile size.
In order to
increase parallelism, we use so called TT kernels which are more parallel but
performs potentially more communication and are less efficient in sequential
than the TS kernels. (A longer discussion on the issue can be found in Section~\ref{sec:kernels}.)
This is another trade-off that we made and we opted for as much parallelism as possible.

We can quote three manuscripts who use some kind of rectangular tiling.  Demmel et
al.~\cite{CAQR} sequentially process rectangular tiles with a recursive QR
factorization algorithm (which is communication optimal in sequential) and then
uses reduction trees to perform the QR factorization in parallel.  Experimental
results are given using a binary tree on tall and skinny matrices. The same algorithms 
is used on the grid (grid of clusters) in~\cite{tsqr-grid}. 
The ScaLAPACK algorithm is used independently on each cluster on a large parallel distributed rectangular tile;
then, a binary tree is used at the grid level among the clusters.
Demmel et
al.~\cite{mohiyuddin:SC09} use a binary tree on top of a flat tree for tall and
skinny matrices. The binary tree is therefore used on rectangular tiles. The flat tree is used locally on the nodes to reduce
sequential communication, while the binary tree is used within the nodes to
increase parallelism.  Finally, the approach of Hadri et
al.~\cite{Hadri_ipdps_2010} is not only interesting in term of parallelism to
tackle various matrix shapes, it is also interesting
in reducing communication (same approach in this case as
in~\cite{mohiyuddin:SC09}) and enables the use of TS kernels.

The sequential kernels of the Tiled QR factorization (executed on a core) are
made of standard blocked algorithms ala LAPACK encoded in kernels; the
development of these kernels is well understood. The focus of this manuscript is on
improving the overall degree of parallelism of the algorithm.  Given a
$p$-by-$q$ tile matrix, we seek to find an appropriate sequence of unitary
transformations on the tiled matrix so as to maximize parallelism (minimize
critical path length). We will get our inspiration in previous work from the
70s/80s on Givens rotations where the question was somewhat related: given an
$m$-by-$n$ matrix, find an appropriate sequence of Givens rotations as to
maximize parallelism. This question is essentially answered
in~\cite{j12,j14,ModiClarke84,SamehKuck78}; we call this class of
algorithms ``coarse-grain algorithms.''

Working with tiles instead of scalars, we introduce four essential differences
between the analysis and the reality of the tiled algorithms and the
coarse-grain algorithms.  First, while there are only two states for a scalar
(nonzero or zero), a tile can be in three states (zero, triangle or full).
Second, there are more operations available on tiles to introduce zeros; we have
a total of three different tasks which can introduce zeros in a matrix.  Third,
the factorization and the update are dissociated to enable factorization stages
to overlap with update stages. In the coarse-grain algorithm, the factorization and the
associated update are considered as a single stage. Fourth and last, while
coarse-grain algorithms have only one task, we end up with six different tasks,
which have different computational weights; this dramatically complicates the critical path analysis of
the tiled algorithms.

While the \Greedy algorithm is optimal for ``coarse-grain algorithms'', we show
that it is not in the case of tiled algorithms. We are unable to devise an
optimal algorithm at this point, but we can prove that both \Greedy and \MC are
asymptotically optimal for all matrices of size $p = q^2 f(q)$, where $f$ is
any function such that $\lim_{+\infty} f= 0$. This result applies to all
matrices where $p$ and $q$ are proportional, $ p = \lambda q$, with $\lambda
\geq 1$, thereby encompassing many important situations in practice (least
squares).

This manuscript is organized as follows. Section~\ref{sec.QR} reviews the numerical
kernels needed to perform a tiled QR factorization, and introduces elimination lists,
which enable us to formally define tiled algorithms. Section~\ref{sec.CP} presents the core
algorithmic contributions of this manuscript. One major result is the asymptotic optimality of two new tiled algorithms, \MC and \Greedy.
Section~\ref{sec.experiments} is devoted to numerical experiments on
multicore platforms. For tall matrices ($p \geq 2q$), these experiments confirm the superiority
of the new algorithms over state-of-the-art solutions of the PLASMA library~\cite{Buttari2008,tileplasma,CAQR,Hadri_ipdps_2010}.
Finally, we provide some concluding remarks in Section~\ref{sec.conclusion}.

\section{The QR factorization algorithm}
\label{sec.QR}

Tiled algorithms are expressed in terms of tile operations rather than
elementary operations.  Each tile is of size $n_b \times n_b$, where $n_b$ is a
parameter tuned to squeeze the most out of arithmetic units and memory
hierarchy. Typically, $n_b$ ranges from $80$ to $200$ on state-of-the-art
machines~\cite{sc09-agullo}. Algorithm~\ref{alg.QR} outlines a naive tiled QR
algorithm, where loop indices represent tiles:

\begin{algorithm}[htbp]
  \DontPrintSemicolon
  \For{$\textnormal{k} = 1$ to $\min(p,q)$}{
     \For{$\textnormal{i} = k+1$ to $p$}{
     $\elim(i, piv(i,k), k)$\;
    }
  }
\caption{Naive QR algorithm for a tiled $p \times q$ matrix.}
\label{alg.QR}
\end{algorithm}

In Algorithm~\ref{alg.QR}, $k$ is the panel index, and $\elim(i, piv(i,k), k)$
is an orthogonal transformation that combines rows $i$ and $piv(i,k)$ to zero
out the tile in position $(i,k)$. However, this formulation is somewhat
misleading, as there is much more freedom for QR factorization algorithms than,
say, for Cholesky algorithms (and contrarily to LU elimination algorithms,
there are no numerical stability issues).  For instance in column $1$, the
algorithm must eliminate all tiles $(i,1)$ where $i>1$, but it can do so in
several ways. Take $p=6$. Algorithm~\ref{alg.QR} uses the transformations
\[\elim(2, 1, 1), \elim(3, 1, 1), \elim(4, 1, 1), \elim(5, 1, 1), \elim(6, 1, 1)\]
But the following scheme is also valid:
\[\elim(3, 1, 1), \elim(6, 4, 1), \elim(2, 1, 1), \elim(5, 4, 1), \elim(4, 1, 1)\]
In this latter scheme, the first two transformations  $\elim(3, 1, 1)$ and $\elim(6, 4, 1)$
use distinct pairs of rows, and they can execute in parallel. On the contrary, $\elim(3, 1, 1)$ and
$\elim(2, 1, 1)$ use the same pivot row and must be sequentialized.
To complicate matters, it is possible to have two orthogonal transformations that execute in parallel
but involve zeroing a tile in two different columns. For instance we can add $\elim(6, 5, 2)$
to the previous transformations and run it concurrently with, say, $\elim(2, 1, 1)$.
Any tiled QR algorithm will be characterized by an \emph{elimination list}, which provides
the ordered list of the transformations used to zero out all the tiles below the diagonal.
This elimination list must obey certain conditions so that the factorization is valid.
For instance, $\elim(6, 5, 2)$ must follow $\elim(6, 4, 1)$ and $\elim(5, 4, 1)$
in the previous list, because there is a flow dependence between these transformations.
Note that, although the elimination list is given as a totally ordered sequence, some transformations can execute
in parallel, provided that they are not linked by a dependence: in the example,
$\elim(6, 4, 1)$ and $\elim(2, 1, 1)$  could have been swapped, and the elimination list
would still be valid.

Before formally stating the conditions that guarantee the validity of (the
elimination list of) an algorithm, we explain how orthogonal transformations
can be implemented.

\subsection{Kernels}
\label{sec:kernels}

To implement a given orthogonal transformation $\elim(i, piv(i,k),
k)$, one can use six different kernels, whose costs are given in
Table~\ref{tab.kernels}.  In this table, the unit of time is the time to
perform $\frac{n_b^3}{3}$ floating-point operations.

\begin{table*}
\centering
\begin{tabular}{|c||c|c||c|c|}
  \hline
  Operation & \multicolumn{2}{|c||}{Panel} & \multicolumn{2}{c|}{Update} \\ \hline
  & Name & Cost & Name & Cost \\ \hline
  Factor square into triangle        & \GEQRT & 4 & \UNMQR & 6  \\ \hline
  Zero square with triangle on top   & \TSQRT & 6 & \TSMQR & 12 \\ \hline
  Zero triangle with triangle on top & \TTQRT & 2 & \TTMQR & 6  \\ \hline
\end{tabular}
\caption{Kernels for tiled QR. The unit of time is $\frac{n_b^3}{3}$ floating-point operations, where $n_b$ is the blocksize.}
\label{tab.kernels}
\end{table*}

There are two main possibilities to implement an orthogonal transformation $\elim(i,
piv(i,k), k)$: The first version eliminates tile $(i,k)$ with the \emph{TS (Triangle
on top of square)} kernels, as shown in Algorithm~\ref{alg.elimSQ}:

\begin{algorithm}[htbp]
  \DontPrintSemicolon
  $\GEQRT(piv(i,k), k)$\;
  $\TSQRT(i,piv(i,k), k)$\;
  \For{$\textnormal{j} = k+1$ to $q$}{
     $\UNMQR(piv(i,k), k, j)$\;
     $\TSMQR(i, piv(i,k), k, j)$\;
  }
\caption{Elimination $\elim(i, piv(i,k), k)$ via \emph{TS (Triangle on top of square)} kernels.}
\label{alg.elimSQ}
\end{algorithm}

Here the tile panel $(piv(i,k), k)$ is factored into a triangle (with \GEQRT).
The transformation is applied to subsequent tiles $(piv(i,k),j)$, $j>k$, in row
$piv(i,k)$ (with \UNMQR). Tile $(i,k)$ is zeroed out (with \TSQRT), and
subsequent tiles $(i,j)$, $j>k$, in row $i$ are updated (with \TSMQR). The flop
count is $4+6+(6+12)(q-k)=10+18(q-k)$ (expressed in same time unit as in
Table~\ref{tab.kernels}).  Dependencies are the following:
\[
\begin{array}{ll}%
\GEQRT(piv(i,k), k) \prec \TSQRT(i, piv(i,k), k)\\%
\GEQRT(piv(i,k), k) \prec \UNMQR(piv(i,k), k, j) & \text{ for } j>k\\%
\UNMQR(piv(i,k) ,k, j) \prec \TSMQR(i, piv(i,k), k, j) & \text{ for } j>k\\%
\TSQRT(i,piv(i,k),k) \prec \TSMQR(i, piv(i,k), k, j) & \text{ for } j>k%
\end{array}%
\]
Note that $\TSQRT(i, piv(i,k), k)$ and $\UNMQR(piv(i,k), k, j)$ can be executed in
parallel, as well as \UNMQR operations on different columns $j, j' >k$. With an
unbounded number of processors, the parallel time is thus $4+6+12=22$
time-units.

The second approach to implement the orthogonal transformation \linebreak $\elim(i,
piv(i,k), k)$ is with the \emph{TT (Triangle on top of triangle)} kernels, as shown in
Algorithm~\ref{alg.elimTR}:

\begin{algorithm}[htbp]
  \DontPrintSemicolon
  $\GEQRT(piv(i,k), k)$\;
  $\GEQRT(i, k)$\;
  \For{$\textnormal{j} = k+1$ to $q$}{
     $\UNMQR(piv(i,k), k, j)$\;
     $\UNMQR(i, k, j)$\;
  }
  $\TTQRT(i, piv(i,k), k)$\;
  \For{$\textnormal{j} = k+1$ to $q$}{
     $\TTMQR(i, piv(i,k), k, j)$\;
  }
\caption{Elimination $\elim(i, piv(i,k), k)$ via \emph{TT (Triangle on top of triangle)} kernels.}
\label{alg.elimTR}
\end{algorithm}

Here both tiles $(piv(i,k),k)$ and $(i,k)$ are factored into a triangle (with
\GEQRT). The corresponding transformations are applied to subsequent tiles
$(piv(i,k),j)$ and $(i,j)$, $j>k$, in both rows $piv(i,k)$ and $i$ (with \UNMQR).
Tile $(i,k)$ is zeroed out (with \TTQRT), and subsequent tiles $(i,j)$, $j>k$,
in row $i$ are updated (with \TTMQR).  The flop count is
$2(4+6(q-k))+2+6(q-k)=10+18(q-k)$, just as before.  Dependencies are the
following:
\[
\begin{array}{ll}%
\GEQRT(piv(i,k), k) \prec \UNMQR(piv(i,k), k, j) & \text{ for } j>k\\%
\GEQRT(i, k) \prec \UNMQR(i, k, j) & \text{ for } j>k\\%
\GEQRT(piv(i,k), k) \prec \TTQRT(i, piv(i,k), k)\\%
\GEQRT(i, k) \prec \TTQRT(i, piv(i,k), k)\\%
\TTQRT(i, piv(i,k), k) \prec \TTMQR(i, piv(i,k), k, j) & \text{ for } j>k\\%
\UNMQR(piv(i,k), k, j) \prec \TTMQR(i, piv(i,k), k, j) & \text{ for } j>k\\%
\UNMQR(i, k, j) \prec \TTMQR(i, piv(i,k), k, j) & \text{ for } j>k%
\end{array}%
\]%
Now the factor operations in row $piv(i,k)$ and $i$ can be executed in parallel.
Moreover, the \UNMQR updates can be run in parallel with the \TTQRT
factorization.  Thus, with an unbounded number of processors, the parallel time
is $4+6+6=16$ time-units.

\medskip
In Algorithm~\ref{alg.elimSQ} and~\ref{alg.elimTR}, it is understood that if a
tile is already in triangle form, then the associated \GEQRT and update kernels
are not applied.

All the new algorithms
introduced in this manuscript are based on \emph{TT} (kernels.  From an algorithmic perspective, \emph{TT} kernels
are more appealing than \emph{TS} kernels, as they offer more parallelism.
More precisely, we can always break a \emph{TS} kernel into two \emph{TT} kernels: We
can replace a $\TSQRT(i,piv(i,k), k)$ (following a $\GEQRT(piv(i,k), k)$) by a $\GEQRT(i, k)$ and a $\TTQRT(i, piv(i,k), k)$. A similar transformation can be made for the updates. Hence a \emph{TS}-based tiled algorithm can always be executed with \emph{TT} kernels,
while the converse is not true.
However, the \emph{TS} kernels provide more data locality, they
benefit form a very efficient implementation (see Section~\ref{sec.experiments}), and several existing
algorithms use these kernels. For all these reasons, and for comprehensiveness,
our experiments will compare approaches based on both kernel types.

Currently (April 2011), the PLASMA library only contains \emph{TS} kernels. We have
mapped the PLASMA algorithm to \emph{TT} kernel algorithm using this
conversion. Going from a \emph{TS} kernel algorithm to a \emph{TT} kernel
algorithm is implicitly done by Hadri et al.~\cite{Hadri_enhancingparallelism} when going
from their ``Semi-Parallel'' to their ``Fully-Parallel'' algorithms.

\subsection{Elimination lists}

As stated above, any algorithm factorizing a tiled matrix
of size $p \times q$ is characterized by its elimination list.
Obviously, the algorithm must zero out all tiles below the diagonal: for each tile $(i,k)$, $i>k$,
$1 \leq k \leq \min(p,q)$, the list must contain exactly one entry $\elim(i, \s, k)$, where \s denotes some
row index $piv(i,k)$ . There are two conditions for a transformation $\elim(i, piv(i,k), k)$ to be valid:
\begin{itemize}
  \item both rows $i$ and $piv(i,k)$ must be ready, meaning that all their tiles
      left of the panel (of indices $(i,k')$ and $(piv(i,k),k')$ for $1 \leq k' < k$)
      must have already been zeroed out: all transformations $\elim(i, piv(i,k'), k')$
       and $\elim(piv(i,k), piv(piv(i,k),k'), k')$ must precede
      $\elim(i, piv(i,k), k)$ in the elimination list
  \item row $piv(i,k)$ must be a potential annihilator, meaning that tile
      $(piv(i,k),k)$ has not been zeroed out yet: \newline the transformation
      $\elim(piv(i,k), piv(piv(i,k),k), k)$ must follow \linebreak $\elim(i, piv(i,k), k)$ in the elimination list
\end{itemize}
Any algorithm that factorizes the tiled
matrix obeying these conditions is called a \emph{generic tiled algorithm} in
the following.

A critical result is that no matter what elimination list is used the total
weight of the tasks for performing a tiled QR factorization algorithm is
constant and equal to $6 p q^2 - 2 q^3$. Using our unit task weight of $n_b^3/3$,
with $m = p n_b$, and $n=q n_b$, we obtain $2 mn^2 -2/3 n^3 $ flops which is the
exact same number as for a standard Householder reflection algorithm as found
in LAPACK (e.g.,~\cite{lawn41}).
We note that this results is true if (a) we use TS kernels as well
and if (b) we use any tiling, (e.g. rectangular tiles).

\subsection{Execution schemes}

In essence, the execution of a generic tiled algorithm is fully determined by
its elimination list.  This list is statically given as input to the scheduler,
and the execution progresses dynamically, with the scheduler executing all
required transformations as soon as possible.  More precisely, each
transformation involves several kernels, whose execution starts as soon as they
are ready, i.e., as soon as all dependencies have been enforced.
Recall that a tile $(i,k)$ can be zeroed out only after all tiles $(i,k')$, with
$k' < k$, have been zeroed out. Execution progresses as follows:
\begin{compactitem}
\item Before being ready for elimination, tile $(i,k)$, $i>k$, must be updated
    $k-1$ times, in order to zero out the $k-1$ tiles to its left (of index
    $(i,k')$, $k' < k$). The last update is a transformation $\TTMQR(i, piv(i,k-1),
    k-1, k)$ for some row index $piv(i,k-1)$ such that $\elim(i, piv(i,k-1), k-1)$
    belongs to the elimination list. When completed, this transformation
    triggers the transformation $\GEQRT(i,k)$, which can be executed
    immediately after the completion of the \TTMQR. In turn, $\GEQRT(i,k)$
    triggers all updates $\UNMQR(i, k, j)$ for all $j > k$.  These updates
    are executed as soon as they are ready for execution.

\item The elimination $\elim(i, piv(i,k), k)$ is performed as soon as possible
    when both rows $i$ and $piv(i,k)$ are ready. Just after the completion of \linebreak
    $\GEQRT(i, k)$ and $\GEQRT(piv(i,k), k)$, kernel
    $\TTQRT(i,$ $piv(i,k), k)$ is launched. When finished, it
    triggers the updates $\TTMQR(i,piv(i,k), k, j)$ for all $j > k$.
\end{compactitem}

Obviously, the degree of parallelism that can be achieved depends upon the
eliminations that are chosen.  For instance, if all eliminations in a given
column use the same factor tile, they will be sequentialized.  This corresponds
to the flat tree elimination scheme described below: in each column $k$, it
uses $\elim(i, k, k)$ for all $i>k$. On the contrary, two eliminations
$\elim(i, piv(i,k), k)$  and $\elim(i', piv(i',k), k)$ in the same column can
be fully parallelized provided that they involve four different rows. Finally,
note that several eliminations can be initiated in different columns
simultaneously, provided that they involve different pairs of rows, and that
all these rows are ready (i.e., they have the desired number of leftmost
zeros).

The following lemma will prove very useful; it states that we can assume
w.l.o.g.\ that each tile is zeroed out by a tile above it, closer to the
diagonal.

\begin{lemma}
\label{th.above}
Any generic tiled algorithm can be modified, without changing its
execution time, so that all eliminations $\elim(i, piv(i,k), k)$ satisfy to $i >
piv(i,k)$.
\end{lemma}

\begin{proof}
    Define a \emph{reverse} elimination as an elimination $\elim(i, piv(i,k), k)$
    where $i < piv(i,k)$.  Consider a generic tiled algorithm whose
    elimination list contains some reverse eliminations. Let $k_0$ be the first
    column to contain one of them. Let $i_{0}$ be the largest row index involved in
    a reverse elimination in column $k_0$.  The elimination list in column $k_0$
    may contain several reverse eliminations  $\elim(i_1, i_0, k_0)$, $\elim(i_2,
    i_0, k_0)$, \dots, $\elim(i_r, i_0, k_0)$, in that order, before row $i_0$ is
    eventually zeroed out by the transformation $\elim(i_0,$ $ piv(i_0,k_0), k_0)$.
    Note that $piv(i_0,k_0) < i_0$ by definition of $i_0$.  We modify
    the algorithm by exchanging the roles of rows $i_0$ and $i_1$ in column $k_0$:
    the elimination list now includes $\elim(i_0, i_1, k_0)$, $\elim(i_2, i_1,
    k_0)$, \dots, $\elim(i_r, i_1, k_0)$, and \newline
    $\elim(i_1, piv(i_0,k_0), k_0)$. All
    dependencies are preserved, and the execution time is unchanged. Now the largest
    row index involved in a reverse elimination in column $k_0$ is strictly smaller
    than $i_0$, and we repeat the procedure until there does not remain any reverse
    elimination in column $k_0$. We proceed inductively to the following columns,
    until all reverse eliminations have been suppressed.
\end{proof}

\section{Critical paths}
\label{sec.CP}

In this section we describe several generic tiled algorithms, and we provide
their critical paths, as well as optimality results. These algorithms are
inspired by algorithms that have been introduced twenty to thirty years
ago~\cite{SamehKuck78,ModiClarke84,j14,j12}, albeit for a much simpler,
\emph{coarse-grain} model.  In this ``old'' model, the time-unit is the time needed to
execute an orthogonal transformation across two matrix rows, regardless of the
position of the zero to be created, hence regardless of the length of these rows.
Although the granularity is much coarser in this model, any existing algorithm
for the old model can be transformed into a generic tiled algorithm, just by
enforcing the very same elimination list provided by the algorithm.
Critical paths are obtained using a discrete event based simulator specially developed
to this end, based on the Simgrid framework~\cite{simgrid}. It carefully
handles dependencies across tiles, and allows for the analysis of both static and dynamic
algorithms.\footnote{The discrete event based simulator, together with the code for all tiled algorithms,
is publicly available at \TiledQRURL}

\subsection{Coarse-grain algorithms}\label{sec:Coarsegrain}

We start with a short description of three algorithms for the coarse-grain model.
These algorithms are illustrated in Table~\ref{tab.coarse} for a $15 \times 6$ matrix.

\paragraph{\SK algorithm}
The \SK algorithm~\cite{SamehKuck78} uses the panel row for all eliminations in
each column, starting from below the diagonal and proceeding downwards.
Time-steps indicate the time-unit at which the elimination can be done,
assuming unbounded resources. Formally, the elimination list is
\[ \Big\{ \Big( \elim(i, k, k), i=k+1, k+2, \dots, p \Big), k =1, 2, \dots, \min(p,q) \Big\} \]

\paragraph{\MC algorithm}
The \MC algorithm is the Fibonacci scheme of order $1$ in~\cite{ModiClarke84}.
Let $\coarse(i,k)$ be the time-step at which tile $(i,k)$, $i>k$, is zeroed
out.  These values are computed as follows. In the first column, there are one
$5$, two $4$'s, three $3$'s, four $2$'s and four $1$'s (we would have had five
$1$'s with $p=16$).  Given $x$ as the least integer such that $x(x+1)/2 \geq
p-1$, we have $\coarse(i,1)= x-y+1$ where $y$ is the least integer such that $
i \leq y(y+1)/2 +1$.  Let the row indices of the $z$ tiles that are zeroed out at
step $s$, $1 \leq s \leq x$, range from $i$ to $i+z-1$. The elimination list
for these tiles is $\elim(i+j, piv(i+j,1), 1)$, with $piv(i+j) = i+j-z$ for $0
\leq j \leq z-1$. In other words, to eliminate a bunch of $z$ consecutive tiles
at the same time-step, the algorithm uses the $z$ rows above them, pairing them
in the natural order.  Now the elimination scheme of the next column is the
same as that of the previous column, shifted down by one row, and adding two
time-units: $\coarse(i,k) = \coarse(i-1,k-1)+2$, while the pairing obeys the
same rule.

\paragraph{\Greedy algorithm}
At each step, the \Greedy algorithm~\cite{j14,j12} eliminates as many tiles as
possible in each column, starting with bottom rows. The pairing for the
eliminations is done exactly as for \MC.  There is no closed-form formula to
compute $\coarse(i,k)$, the time-step at which tile $(i,k)$ is eliminated, but
it is possible to provide recursive expressions (see~\cite{j14,j12}).

\begin{table}[htbp]
\centering
\resizebox{0.6\linewidth}{!}{%
\begin{tabular}{|rrrrrr|rrrrrr|rrrrrr|}%
\hline%
\multicolumn{6}{|c|}{ (a) \SK } & \multicolumn{6}{c|}{ (b) \MC }  &\multicolumn{6}{c|}{ (c) \Greedy } \\%
\hline%
\s &    &    &    &    &          &\s  &      &     &     &      &        &    \s  &     &    &     &     &    \\%
  1 & \s &    &    &    &         & 5  & \s   &     &     &      &        &     4  & \s  &    &     &     &    \\%
  2 &   3 & \s &    &    &        & 4  &  7   & \s  &     &      &        &     3  &  6  & \s &     &     &    \\%
   3 &   4 &   5 & \s &    &      & 4  &  6   & 9  &  \s &      &         &     3  &  5  &  8 &  \s &     &    \\%
   4 &   5 &   6 &   7 & \s &     & 3  &  6   & 8  &  11 &   \s &         &     2  &  5  &  7 &  10 &  \s &    \\%
   5 &   6 &   7 &   8 &   9 & \s & 3  &  5   & 8  &  10 &   13 &   \s    &     2  &  4  &  7 &   9 &  12 &  \s\\%
   6 &   7 &   8 &   9 &  10 &  11& 3  &  5   & 7  &  10 &   12 &   15    &     2  &  4  &  6 &   9 &  11 &  14\\%
   7 &   8 &   9 &  10 &  11 &  12& 2  &  5   & 7  &   9 &   12 &   14    &     2  &  4  &  6 &   8 &  10 &  13\\%
   8 &   9 &  10 &  11 &  12 &  13& 2  &  4   & 7  &   9 &   11 &   14    &     1  &  3  &  5 &   8 &  10 &  12\\%
   9 &  10 &  11 &  12 &  13 &  14& 2  &  4   & 6  &   9 &   11 &   13    &     1  &  3  &  5 &   7 &   9 &  11\\%
  10 &  11 &  12 &  13 &  14 &  15& 2  &  4   & 6  &   8 &   11 &   13    &     1  &  3  &  5 &   7 &   9 &  11\\%
  11 &  12 &  13 &  14 &  15 &  16& 1  &  4   & 6  &   8 &   10 &   13    &     1  &  3  &  4 &   6 &   8 &  10\\%
  12 &  13 &  14 &  15 &  16 &  17& 1  &  3   & 6  &   8 &   10 &   12    &     1  &  2  &  4 &   6 &   8 &  10\\%
  13 &  14 &  15 &  16 &  17 &  18& 1  &  3   & 5  &   8 &   10 &   12    &     1  &  2  &  4 &   5 &   7 &   9\\%
  14 &  15 &  16 &  17 &  18 &  19& 1  &  3   & 5  &   7 &   10 &   12    &     1  &  2  &  3 &   5 &   6 &   8\\%
\hline%
\end{tabular}%
}
  \caption{Time-steps for coarse-grain algorithms.}
  \label{tab.coarse}
\end{table}

Consider a rectangular $p \times q$ matrix, with $p>q$.  With the coarse-grain
model, the critical path of \SK is $p+q-2$, and that of \MC is $x + 2q - 2$,
where $x$ is the least integer such that $x(x+1)/2 \geq p-1$.  The critical
path of \Greedy is unknown, but two important results are known: (i) the
critical path of \Greedy is optimal; (ii) its value tends to $2q$ if $p$ is
negligible in front of $q^2$, i.e., if we have $p = q^2 f(q)$ where $f$ is any
function such that $\lim_{+\infty} f= 0$ (and $f(q) > 1/q$ so that $p > q$).
In particular, let $p$ and $q$ be proportional, $p = \lambda q$, with a
constant $\lambda > 1$: \MC is asymptotically optimal, because $x$ is of the
order of $\sqrt{q}$, hence its critical path is $2q + o(q)$. On the contrary,
\SK is not asymptotically optimal since its critical path is $(1+\lambda)q -
2$.  For square $q \times q$ matrices, critical paths are slightly different
($2q-3$ for \SK, $x+2q-4$ for \MC), but the important result is that all three
algorithms are asymptotically optimal in that case.

\begin{algorithm}[h]
\scriptsize{
   \DontPrintSemicolon
   \For{ $j = 1$ to $q$}{
        \tcc{$nz(j)$ is the number of tiles which have been eliminated in column $j$}
        $nZ(j) = 0$\;
        \tcc{$nT(j)$ is the number of tiles which have been triangularized in column $j$}
        $nT(j) = 0$\;
    }

\While{ column $q$ is not finished}{

   \For{ $j = q$ down to $1$}{

        \If{$j==1$}{
            \tcc{Triangularize the first column if not yet done}
            $nT_{\textnormal{new}} = nT(j) + ( p - nT(j) )$\;
            \If{ $p-nT(j) > 0$ }{
                \For{$k = p$ down to $1$ }{
                    $\GEQRT(k, j)$\;
                    \For{$\textnormal{jj} = j+1$ to $q$}{
                        $\UNMQR(k, j, jj)$\;
                     }
                }
            }
        }
        \Else{
            \tcc{Triangularize every tile having a zero in the previous column}
            $nT_{\textnormal{new}} = nZ(j-1)$\;
            \For{$k = nT(j)$ to $nT_{\textnormal{new}}-1$}{
                $\GEQRT(p-k, j)$\;
                \For{$\textnormal{jj} = j+1$ to $q$}{
                  $\UNMQR(p-k, j, jj)$\;
                }
            }
        }

        \tcc{Eliminate every tile triangularized in the previous step}
        $nZ_{\textnormal{new}} = nZ(j) + \lfloor \dfrac{ nT(j) - nZ(j) }{2} \rfloor$\;
        \For{ $kk = nZ(j)$ to $nZ_{\textnormal{new}}-1$}{
            $piv(p-kk) = p-kk-nZ_{\textnormal{new}}+nZ(j)$\;
            $\TTQRT(p-kk, piv(p-kk), j)$\;
            \For{$\textnormal{jj} = j+1$ to $q$}{
                $\TTMQR(p-kk, piv(p-kk), j, jj)$\\
            }
        }

        \tcc{Update the number of triangularized and eliminated tiles at the next step}
        $nT(j) = nT_{\textnormal{new}}$\;
        $nZ(j) = nZ_{\textnormal{new}}$\;
    }
}
}
\caption{\Greedy algorithm via \emph{TT} kernels.}
\label{alg.tiled-greedy}
\end{algorithm}


\begin{table*}[tb]
\centering
\resizebox{.8\linewidth}{!}{%
\begin{tabular}{|rrrrrr|rrrrrr|rrrrrr|rrrrrr|rrrrrr|}%
\hline
\multicolumn{6}{|c|}{ (a) \SK }       & \multicolumn{6}{c|}{ (b) \MC }        &\multicolumn{6}{c|}{ (c) \Greedy }& \multicolumn{6}{c|}{ (d) \BT }  & \multicolumn{6}{c|}{ (e) \PT ($\BS = 5$) }  \\%
\hline%
 \s  &     &     &     &      &     & \s  &     &     &     &     &        &  \s&    &    &    &    &        &  \s &     &     &     &     &    & \s &    &    &    &    &   \\%
  6  & \s  &     &     &      &     &  14 & \s  &     &     &     &        &  12&  \s&    &    &    &        &   6 & \s  &     &     &     &    &  6 &  \s &    &    &    &   \\%
  8  &  28 & \s  &     &      &     &  12 &  48 & \s  &     &     &        &  10&  42&  \s&    &    &        &   8 &  28 & \s  &     &     &    &  8 &  28 &  \s &    &    &   \\%
 10  &  34 &  50 &  \s &      &     &  12 &  46 &  70 & \s  &     &        &  10&  40&  64&  \s&    &        &   6 &  36 &  56 &  \s &     &    & 10 &  34 &  50 &  \s &    &   \\%
 12  &  40 &  56 &  72 &  \s  &     &  10 &  42 &  68 &  92 & \s  &        &   8&  36&  62&  86& \s &        &  10 &  34 &  70 &  90 & \s  &    & 12 &  40 &  56 &  72 &  \s &   \\%
 14  &  46 &  62 &  78 &  94  &   \s&  10 &  40 &  64 &  90 & 114 & \s     &   8&  34&  56&  84& 106& \s     &   6 &  44 &  68 & 104 & 124 & \s & 14 &  46 &  62 &  78 &  94 &  \s\\%
 16  &  52 &  68 &  84 & 100  &  116&  10 &  40 &  62 &  86 & 112 & 136    &   8&  34&  56&  78& 102& 128    &   8 &  28 &  78 & 102 & 138 & 158&  6 &  54 &  74 &  90 & 106 & 122\\%
 18  &  58 &  74 &  90 & 106  &  122&   8 &  36 &  62 &  84 & 108 & 134    &   8&  30&  52&  78& 100& 122    &   6 &  42 &  62 & 112 & 136 & 172&  8 &  28 &  82 & 102 & 118 & 134\\%
 20  &  64 &  80 &  96 & 112  &  128&   8 &  34 &  58 &  84 & 106 & 130    &   6&  28&  50&  72& 100& 118    &  12 &  40 &  76 &  96 & 146 & 170& 10 &  34 &  50 & 110 & 130 & 146\\%
 22  &  70 &  86 & 102 & 118  &  134&   8 &  34 &  56 &  80 & 106 & 128    &   6&  28&  50&  72&  94& 116    &   6 &  46 &  74 & 110 & 130 & 180& 12 &  40 &  56 &  72 & 138 & 158\\%
 24  &  76 &  92 & 108 & 124  &  140&   8 &  34 &  56 &  78 & 102 & 128    &   6&  28&  50&  68&  94& 116    &   8 &  28 &  80 & 108 & 144 & 164& 16 &  52 &  68 &  84 & 100 & 166\\%
 26  &  82 &  98 & 114 & 130  &  146&   6 &  28 &  56 &  78 & 100 & 122    &   6&  28&  44&  66&  88& 110    &   6 &  36 &  56 & 114 & 142 & 178&  6 &  56 &  80 &  96 & 112 & 128\\%
 28  &  88 & 104 & 120 & 136  &  152&   6 &  28 &  50 &  78 & 100 & 122    &   6&  22&  44&  66&  88& 110    &  10 &  34 &  64 &  84 & 148 & 176&  8 &  28 &  84 & 108 & 124 & 140\\%
 30  &  94 & 110 & 126 & 142  &  158&   6 &  28 &  44 &  72 & 100 & 122    &   6&  22&  44&  60&  82& 104    &   6 &  38 &  62 &  92 & 112 & 182& 10 &  34 &  50 & 112 & 136 & 152\\%
 32  & 100 & 116 & 132 & 148  &  164&   6 &  22 &  44 &  60 &  94 & 116    &   6&  22&  38&  60&  76&  98    &   8 &  28 &  66 &  90 & 114 & 134& 12 &  40 &  56 &  72 & 140 & 164\\%
\hline
\end{tabular}%
}
  \caption{Time-steps for tiled algorithms.}
  \label{tab.tiled}
\end{table*}


\subsection{Tiled algorithms}\label{sec:TiledAlgorithms}

As stated above, each coarse-grain algorithm can be transformed into a tiled
algorithm, simply by keeping the same elimination list, and triggering the
execution of each kernel as soon as possible.
However, because the weights of
the factor and update kernels are not the same, it is much more difficult to
compute the critical paths of the transformed (tiled) algorithms.
Table~\ref{tab.tiled} is the counterpart of Table~\ref{tab.coarse}, and depicts
the time-steps at which tiles are actually zeroed out. Note that the tiled
version of \SK is indeed the \FT algorithm in PLASMA~\cite{Buttari2008,tileplasma}, and we have
renamed it accordingly.  As an example, Algorithm~\ref{alg.tiled-greedy} shows
the \Greedy algorithm for the tiled model.

A first (and quite unexpected) result is that \Greedy is no longer optimal, as
shown in the first two columns of Table~\ref{table.15x2x3} for a $15\times2$ matrix.
In each column and at each step, ``{\em the
\ASAP algorithm}'' starts the elimination of a tile as soon as there are at least
two rows ready for the transformation.  When $s \geq 2$ eliminations can start
simultaneously, \ASAP pairs the $2s$ rows just as \MC and \Greedy, the first
row (closest to the diagonal) with row $s+1$, the second row with row $s+2$,
and so on.  As a matter of a fact, when processing the second column, both
\ASAP and \Greedy begin with the elimination of lines 10 to 15 (at time step
20).  However, once tiles $(13,2)$, $(14,2)$ and $(15,2)$ are zeroed out (i.e.
at time step 22), \ASAP eliminates $4$ zeros, in rows $9$ through $12$. On the
contrary, \Greedy waits until time step $26$ to eliminate 6 zeros in rows $6$
through $12$.  In a sense, \ASAP is the counterpart of \Greedy at the tile
level.  However, \ASAP is not optimal either, as shown in
Table~\ref{table.15x2x3} for a $15\times3$ matrix.  On larger examples, the critical path of \Greedy is
better than that of \ASAP, as shown in Table~\ref{table.biggermat}.

We have seen that, for a $15\times 2$ matrix, \ASAP is better than \Greedy and
that, for a $15\times 3$ matrix, \Greedy is better than \ASAP. We can further
improve upon \Greedy in the $15\times 3$ case. We consider the \GRASAP($k$)
algorithm defined as: following the \Greedy algorithm up from columns $1$ to $q-k$
and then switching in \ASAP mode for the last $k$ columns. \GRASAP(0) is
\Greedy, while \GRASAP($q$) is \ASAP.  In
Table~\ref{table.15x2x3}(c), we give the results for \GRASAP(1). In this case
(a $15\times 3$ matrix), \GRASAP(1) is better than \Greedy.  \GRASAP(1) finishes at time-step
$62$, while \Greedy finishes at time-step $64$. Of course it would be interesting to 
determine the best value of $k$ as a function of $p$ and $q$, for the execution of \GRASAP($k$) on a $p \times q$ matrix.

\begin{table}[htbp]
\centering
\resizebox{\linewidth}{!}{%
\subfloat[\Greedy nor \ASAP are optimal.\label{table.15x2x3}]{%
\begin{tabular}{c|rrr|rrr|rrr|c}%
\cline{2-10}%
~~ & \multicolumn{3}{|c|}{ (a)  \Greedy  }  &\multicolumn{3}{c|}{ (b) \ASAP } &\multicolumn{3}{c|}{ (c) \GRASAP(1)}& ~~ \\%
\cline{2-10}%
~~ &                \s&    &                &         \s  &     &             &         \s  &     &             & ~~ \\%
~~ &                12&  \s&                &         12  & \s  &             &         12  & \s  &             & ~~ \\%
~~ &                10&  42&  \s            &         10  & 40  & \s          &         10  & 42  & \s          & ~~ \\%
~~ &                10&  40&  64            &         10  & 36  & 86          &         10  & 40  & 62          & ~~ \\%
~~ &                 8&  36&  62            &          8  & 34  & 80          &          8  & 36  & 58          & ~~ \\%
~~ &                 8&  34&  56            &          8  & 32  & 74          &          8  & 34  & 56          & ~~ \\%
~~ &                 8&  34&  56            &          8  & 30  & 68          &          8  & 34  & 56          & ~~ \\%
~~ &                 8&  30&  52            &          8  & 28  & 62          &          8  & 30  & 50          & ~~ \\%
~~ &                 6&  28&  50            &          6  & 28  & 56          &          6  & 28  & 50          & ~~ \\%
~~ &                 6&  28&  50            &          6  & 26  & 50          &          6  & 28  & 48          & ~~ \\%
~~ &                 6&  28&  50            &          6  & 24  & 46          &          6  & 28  & 46          & ~~ \\%
~~ &                 6&  28&  44            &          6  & 24  & 44          &          6  & 28  & 44          & ~~ \\%
~~ &                 6&  22&  44            &          6  & 22  & 44          &          6  & 22  & 44          & ~~ \\%
~~ &                 6&  22&  44            &          6  & 22  & 40          &          6  & 22  & 40          & ~~ \\%
~~ &                 6&  22&  38            &          6  & 22  & 38          &          6  & 22  & 38          & ~~ \\%
\cline{2-10}%
\end{tabular}%
}
~~
\subfloat[\Greedy generally outperforms \ASAP.]{%
\label{table.biggermat}%
\begin{tabular}{|c|l|r|r|r|r|}%
\cline{3-6}%
\multicolumn{2}{c}{}             & \multicolumn{4}{|c|}{$q$}\\%
\hline%
$p$                  & \multicolumn{1}{c|}{Algorithm} & \multicolumn{1}{c|}{16}  & \multicolumn{1}{c|}{32} & \multicolumn{1}{c|}{64} & \multicolumn{1}{c|}{128}    \\%
\hline%
\multirow{2}{*}{16}  & \Greedy   & 310 & \multirow{2}{*}{~}   &  \multirow{4}{*}{~}  &  \multirow{6}{*}{~}      \\%
                     & \ASAP     & 310 &    &    &        \\%
\cline{1-4}%
\multirow{2}{*}{32}  & \Greedy   & 360 &650 &    &        \\%
                     & \ASAP     & 402 &656 &    &        \\%
\cline{1-5}%
\multirow{2}{*}{64}  & \Greedy   & 374 &726 & 1342&        \\%
                     & \ASAP     & 588 &844 & 1354&        \\%
\hline%
\multirow{2}{*}{128} & \Greedy   & 396 &748 & 1452 & 2732    \\%
                     & \ASAP     & 966 &1222 & 1748& 2756    \\%
\hline%
\end{tabular}%
}
}
\caption{Neither \Greedy nor \ASAP are optimal.\label{table.greedyvsasap}}
\end{table}

\medskip
We have a closed-form formula for the critical path of tiled \FT, but not for
that of tiled \MC (contrarily to the coarse-grain case). But we provide an
asymptotic expression, both for \MC and for \Greedy. More importantly, we show that both
tiled algorithms are asymptotically optimal. We state our main result:

\begin{theorem}
\label{th.main}
For a tiled matrix of size $p \times q$, where $p \geq q$:
\begin{compactenum}
\item The critical path length of \FT is
     \begin{flalign*}
          2p + 2 \quad &\textrm{if $p \geq q = 1$}\\
          6p + 16q - 22 \quad &\textrm{if $p > q > 1$}\\
          22p - 24 \quad &\textrm{if $p = q > 1$}
      \end{flalign*}
\item The critical path length of \MC is at most $22q + 6 \lceil \sqrt{2p} \rceil$, and that of \Greedy is at most $22q + 6 \lceil \log_2{p} \rceil$
\item The optimal critical path length is at least $22q-30$
\item \MC is asymptotically optimal if $p = q^2 f(q)$, where $\lim_{+ \infty} f = 0$
\item \Greedy is asymptotically optimal if $\log_2 p = q f(q)$, where $\lim_{+ \infty} f = 0$
\end{compactenum}
\end{theorem}

\begin{proof}
\textbf{Proof of (1).}
    Consider first the case $p \geq q = 1$. We shall proceed by
    induction on $p$ to show that the critical path of \FT is of length $2p+2$,
    If $p=1$, then from Table~\ref{tab.kernels} the result
    is obtained since only $\GEQRT(1,1)$ is required.  With the base case
    established, now assume that this holds for all $p-1 > q = 1$.  Thus at
    time $t=2(p-1) + 2 = 2p$, we have that for all $p-1 \geq i \geq 1$ tile
    $(i,1)$ has been factorized into a triangle and for all $p-1 \geq i > 1$,
    tile $(i,1)$ has been zeroed out.  Therefore, tile $(p,1)$ will be zeroed
    out with $\TTQRT(p,1)$ at time $t+2 = 2(p-1) + 2 + 2 = 2p + 2$.

    Consider now the case $p > q >1$. We show by induction on $k$ that tile $(i,k)$, for $i>k\geq2$, is zeroed
    out in \FT at time unit $6i + 16k - 22$.  For $k=2$, tile $(2,2)$ is updated from step $k=1$
    at time $4+6+6=16$, and it is factored into a triangle at time $20$. Tile
    $(3,2)$ is updated from step $k=1$ at time $22$ factored into a triangle at
    time $26$ and then zeroed out at time $26+2=28= 6\times 3+16 \times 2 -
    22$. A new tile in column $2$ is zeroed out every $6$ time units, hence the
    initialization of the induction for $k=2$.
    Assume now that the formula holds up to column $k$, and let $t=6(k+1) + 16k
    -22$ be the time at which tile $(k+1,k)$ is zeroed out.  Tile $(k+1,k+1)$
    is updated from step $k$ at time $t-2+6+6=t+10$ and factored into a
    triangle at time $t+14$. By induction, tile $(k+2,k)$ is zeroed out at time
    $t+6$, hence triangularized at time $t+4$. The corresponding $\UNMQR$
    update of tile $(k+2,k+1)$ ends at time $t+10$, its $\TTMQR$ update ends at
    time $\max(t+14,t+10)+6=t+20$.  Hence tile $(k+2,k+1)$ can indeed be zeroed
    out at time $\max(t+12,t+20)+2=t+22$. A new tile in column $k+1$ can be
    zeroed out every $6$ time units, hence the induction formula for $k+1$.

    Finally, for a square matrix of size $q \times q$, consider the above formula for a
    rectangular matrix with $p=q+1$. Instead of zeroing out the last tile
    $(q+1,q)$ with tile $(q,q)$,
    simply need to factor tile $(q,q)$ into a triangle with $\GEQRT(q,q)$. This
    costs $4$ time units instead of $6$ when adding $\TTQRT(q+1,q,q)$, and explains the difference of $2$ in
    the formula for square matrices.

\textbf{Proof of (2).}
    \MC and \Greedy are more difficult to analyze than \FT, but we provide an
    upper bound of their critical path. The approach is the same for both
    algorithms, and hereafter \Alg denotes either \MC or \Greedy.  Let
    $\coarse(i,k)$ be the time-step at which tile $(i,k)$ is zeroed out in \Alg
    with the coarse-grain model (see Table~\ref{tab.coarse} for examples).  We
    derive a ``slowed down'' version of the tiled version of \Alg by
    terminating the zeroing out of tile $(i,k)$ at time-step
    $$6 \coarse(2,1) +
    22(k-1) - 6 (\coarse(k,k) - \coarse(i,k)).$$
    We say that this version is
    slowed down because we do not start the zeroing out of the tiles as soon as
    possible. For instance in the first column, tile $(i,1)$ is zeroed out at
    time $6 \coarse(i,k)$, which is larger than the value given in
    Table~\ref{tab.tiled}.  However, we keep the same elimination list as in
    the original version of \Alg, and we trigger the update and factor
    operations as soon as possible when the zeroing out operation is completed.
    We only delay these latter operations.

    The intuitive idea for delaying the eliminations is that the corresponding
    updates will be fully overlapped, within a given column, or when proceeding
    from one column to the next: in this case, allowing for a time-shift of
    $22$ smooths the chaining of the updates. The regular and repetitive
    spacing of the eliminations allows us to check (just as we did to prove (1))
    that all dependencies are enforced in the slowed down version of \Alg.
    Because the case-analysis is tedious, we have written a program for a
    sanity check of the validity of \Alg\footnote{All program sources are
    publicly available at \TiledQRURL}.

    In the coarse-grain model, \Alg  terminates the first column in time $x$,
    so the critical path of its slowed down version is $6x + 22(q-1)$. For \MC,
    $x$ is the least integer such that $x(x+1)/2 \geq  p-1$, hence $x \leq
    \lceil \sqrt{2p} \rceil$.  For \Greedy, $x = \lceil \log_2(p-1) \rceil \leq
    \lceil \log_2{p} \rceil$, hence the result.

\textbf{Proof of (3).}
    Consider a square $q \times q$ matrix $B$, with $q \geq 2$. Assume that there
    are only three non-zero sub-diagonals, i.e., that tile $(i,k)$ is initially
    zero in $B$ for $i > k+3$. Because there are only three non-zero tiles below the
    diagonal, there is a constant number of possible row pairings in each
    column. An exhaustive search is to try all possible pairings in the first
    column, followed by all possible pairings in the second column, and so on.
    After a few columns, a pattern emerges, and we can identify that any
    optimal algorithm (there are several of them) needs at least $22$
    time-steps to proceed from one column to the next.  It is possible to save
    a few steps at the beginning and end of the execution, and the optimal
    critical path is $22q - 30$. Here also, because the case-analysis is long
    and tedious, we have written a program for a sanity check of the latter
    value.

    Now we show that the optimal critical path for a general $p \times q$
    matrix $A$, with $p \geq q \geq 2$, is at least equal to the critical path
    of the previous $q \times q$ matrix $B$ with three sub-diagonals. Indeed,
    Lemma~\ref{th.above} shows that there exist optimal algorithms for
    factoring $A$ without any reverse elimination. Consider such an algorithm,
    and discard all eliminations that involve zeroing out elements below the
    third sub-diagonal, or outside the $q \times q$ top square: the critical
    path cannot increase, and we have an elimination scheme for $B$, which
    proves the desired result.

    Note that using $B$ instead of $A$ is the key to the proof: in each column of
    $B$, there is only a constant number of possible row pairings, which makes it possible to
    try all combinations for several consecutive columns. Reasoning with $A$ instead would
    need a completely different proof (yet to be invented).

\textbf{Proof of (4) and (5).}
    These are a direct consequence of (3) and (4).
\end{proof}

\noindent
\textbf{Remarks:}
\begin{enumerate}

\item We express all critical path lengths in terms of $p$ and $q$, with an unit of
$n_b^3/3$ floating-point operations. It is easy to get critical path
lengths in term of $m$, $n$, and $n_b$, and with elementary floating-point
operations as unit, assuming that all tiles are full.  (In other words, $m$ and $n$ are
multiple of $n_b$.) For example for \FT, we get
$(2/3) mn_b^2 + (2/3) n_b^3$ if $m \geq n=n_b$,
$2mn_b^2 + 16/3nn_b^2 - (22/3)n_b^3$ if $m > n >n_b $
and $(22/3) nn_b^2 - (24/3)n_b^3$ if $m=n > n_b$.

\item From this formula, it is clearer that, if one wants to minimize the number
of floating-point operations on the critical path, one needs to take $n_b=1$.
However, such an action would have disastrous consequences. The communication
increase would be way too high, and the increase gain in parallelism would not be
worth the overhead. More importantly, the efficiency of the elimination kernels would 
be much lower.  In this manuscript, we consider $n_b$ constant, large
enough so that elimination kernels operate at full Level 3 BLAS performance, and so that
communication costs remain relatively low.

\item In the square case, we see that the critical path length of the tiled
algorithms is typically in $\mathcal{O}(n n_b^2)$. This is in sharp contrast
with the current LAPACK algorithm \GEQRF. If we assume that the panel is not
parallelizable, and that the block size for the LAPACK algorithm is $n_b$, then
counting the length of the chain of panel factorization steps leads to a
critical path length in $\mathcal{O}(n^2 n_b)$. There is therefore much more
parallelism to exploit in the tiled algorithms than in the current LAPACK
algorithms. Or, stated differently~\cite{Buttari2008,tileplasma},
the granularity of the tiled algorithms is finer than that of the LAPACK algorithm.

\end{enumerate}

In Table~\ref{tab.tiled} we also report time-steps for the \BT algorithm.
As its name indicates, this algorithm performs a binary tree
reduction to zero out tiles in each column. Here is an asymptotic expression of
its critical path:

\begin{proposition}
\label{th.bt}
Consider a tiled matrix of size $p \times q$, where $p \geq q$. The critical
path length of \BT is $6 q \log_2 p + o(q \log_2 p)$.
\end{proposition}

\begin{proof}
It is possible to derive an exact expression for the critical path length of
\BT in the special case where $p$ and $q$ are both exact powers of two, with
$q<p$. We obtain the value $(10 + 6 \log_2 p) q - 4 \log_2 p - 6$.  As before,
the proof goes by (tedious) induction. Here again, we have written a program for
a sanity check of the latter value. The asymptotic value follows easily for an
arbitrary matrix, by enlarging each dimension to the nearest power of two.
\end{proof}

Proposition~\ref{th.bt} shows that \BT is not asymptotically optimal.
The PLASMA library provides more
algorithms, that can be informally described as trade-offs between \FT and \BT.
(We remind the reader that \FT is the same as algorithm as \SK.)
These algorithms are referred to as \PT in all the following, and differ by the
value of an input parameter called the \emph{domain size} $\BS$.
This domain size can be any value between $1$ and
$p$, inclusive.  Within a domain, that includes $\BS$ consecutive rows,
the algorithm works just as \FT: the
first row of each domain acts as a local panel and is used to zero out the
tiles in all the other rows of the domain. Then the domains are merged: the panel rows
are zeroed out by a binary tree reduction, just as in \BT.
As the algorithm progresses through the columns, the
domain on the very bottom is reduced accordingly, until such time that there is
one less domain.  For the case that $\BS=1$, \PT follows a
binary tree on the entire column, and for $\BS = p$, the algorithm executes a
flat tree on the entire column.  It seems very
difficult for a user to select the domain size $\BS$ leading to best performance,
but it is known that $\BS$ should increase as $q$ increases. Table~\ref{tab.tiled}
shows the time-steps of \PT with a domain size of $\BS =5$.
In the experiments of
Section~\ref{sec.experiments}, we use all possible values of $\BS$ and retain
the one leading to the best value.

So far our study has only been concerned with algorithms based on TT kernels.
Indeed, in the manuscript, \FT stands for \TTFT.  We now give the critical path
of the algorithm \TSFT. This corresponds to the \FT algorithm (i.e., \SK) with
TS kernels. This algorithm was introduced
in~\cite{Buttari2008,tileplasma,Quintana:2009} and is available in PLASMA for
performing the QR factorization of a matrix on multicore architecture.

\begin{proposition} The critical path length for \TSFT is
\label{prop:tsft}
      \begin{flalign*}
          6p - 2 \quad &\textrm{for  $p \geq q = 1$}\\
          12p + 18q - 32 \quad &\textrm{for $p > q > 1$} \\
          30p - 34 \quad &\textrm{for $p = q > 1$}
      \end{flalign*}
\end{proposition}

\begin{proof}
    Consider the case of $p \geq q = 1$. In order to show that for any $p$,
    with $q=1$, the critical path is of length $6p-2$, we shall proceed by
    induction on $p$.  If $p=q=1$, then from Table~\ref{tab.kernels} the result
    is obtained since only $\GEQRT(1,1)$ is required.  With the base case
    established, now assume that this holds for all $p-1 > q = 1$.  Thus at
    time $t=6(p-1) -2 = 6p - 8$, we have that tile $(1,1)$ has been factorized
    into a triangle and for all $p-1 \geq i > 1$, tile $(i,1)$ has been zeroed
    out.  Therefore, tile $(p,1)$ will be zeroed out with $\TSQRT(p,1)$ at time
    $t+6 = 6(p-1) - 2 + 6 = 6p-2$.

    Assume that $p>q>1$. We show by induction on $k$ that tile $(i,k)$, for
    $i>k\geq2$, is zeroed out at time unit $12i + 18k - 32$.  Tile $(2,2)$ is
    updated from step $k=1$ at time $6(2) - 2 + 12 = 22$, it is factored into a
    triangle at time $28$. Tile $(3,2)$ is zeroed out at time $28+12=40=12
    \times 3+18 \times 2 - 32$, and a new tile in column $2$ is zeroed out
    every $12$ time units, hence the initialization of the induction for $k=2$.

    Assume now that the formula holds up to column $k$, and let $t=12(k+1) +
    18k - 32$ be the time at which tile $(k+1,k)$ is zeroed out.  Tile
    $(k+1,k+1)$ is updated from step $k$ at time $t+12$ and factored into a
    triangle at time $t+18$. By induction, tile $(k+2,k)$ is zeroed out at time
    $t+12$.  Hence tile $(k+2,k+1)$ can indeed be zeroed out at time
    $\max(t+12,t+18)+12=t+30$. A new tile in column $k+1$ can be zeroed out
    every $12$ time units, hence the induction formula for $k+1$.

    For a square matrix of size $q \times q$, consider the above formula for a
    rectangular matrix with $p=q+1$. Instead of zeroing out the last tile
    $(q+1,q)$ with tile $(q,q)$ in $6$ time units with $\TSQRT(q+1,q)$, we
    simply need to factor tile $(q,q)$ into a triangle with $\GEQRT(q,q)$. This
    costs $4$ time units instead of $6$, and explains the difference of $2$ in
    the formula for square matrices.
\end{proof}

As we can see, the critical path of \TSFT (Proposition~\ref{prop:tsft})
is longer than the one of \FT (Theorem~\ref{th.main}(1)).
This stems from the facts that (1) a TS algorithm can be converted into a TT algorithm, and (2)
this conversion increases the parallelism, and, consequently, reduces the critical path length.

\section{Experimental results}
\label{sec.experiments}

All experiments were performed on a 48-core machine composed of eight
hexa-core AMD Opteron 8439 SE (codename Istanbul) processors running at 2.8
GHz. Each core has a theoretical peak of 11.2 Gflop/s with a peak of 537.6
Gflop/s for the whole machine. The Istanbul micro-architecture is a NUMA
architecture where each socket has 6 MB of level-3 cache and each processor has
a 512 KB level-2 cache and a 128 KB level-1 cache.  After having benchmarked
the AMD ACML and Intel MKL BLAS libraries, we selected MKL (10.2) since it
appeared to be slightly faster in our experimental context.  Linux 2.6.32 and
Intel Compilers 11.1 were also used in conjunction with PLASMA 2.3.1.

For all results, we show both double and double complex precision, using all 48
cores of the machine.  The matrices are of size $m=8000$ and $200 \leq n \leq
8000$.  The tile size is kept constant at $n_b=200$, so that the
matrices can also be viewed as $p \times q$ tiled matrices where $p=40$ and $1
\leq q \leq 40$. All kernels use an inner blocking parameter of $i_b=32$.

Asymptotically all operations in a QR factorization are FMAs (``{\em fused
multiply-add}'', $ y \leftarrow \alpha x + y$).  In real arithmetic, an FMA
involves three double precision numbers for two flops, but these two flops can
be combined into one FMA instruction and thus completed in one cycle.  In
complex arithmetic, the operation $ y \leftarrow \alpha x + y$ involves
six double precision numbers for eight flops; we also note that there is no
such thing as a complex-arithmetic FMA.  The ratio of computation/communication
is therefore, potentially, four times higher in complex arithmetic than
in real arithmetic. Communication aware algorithms are much more critical in
real arithmetic than in complex arithmetic. This is the reason why we present
results in complex arithmetic and in real arithmetic. Our new algorithms will
be at their best in the complex arithmetic case where parallelism is most
important while communication less. In the real arithmetic case, we will see
that TS kernels which perform potentially less communication than TT kernels
have the advantage as soon as there is enough parallelism from the algorithm
($q$ large enough).

The PLASMA interface allows one to specify the dependencies between tasks by
designating the data as either INPUT, OUTPUT, INOUT, or NODEP.  Currently,
the update kernels (\UNMQR, \TTMQR, and
\TSMQR) introduced false dependencies between the tasks which sequentializes
the execution of update with the factorization kernels \TTQRT or \TSQRT.  In
order to alleviate these, we altered the dependency designation within each of
the update kernels for the matrix of Householder reflectors, V, from INPUT to
NODEP as is further explained in \cite{CPE:CPE1467}.  The dependencies between
the tasks are still consistent since the T matrix within each update kernel
continues to be designated as INPUT so that any subsequent task which
overwrites this T matrix cannot be executed.

For each experiment, we provide a comparison of the theoretical performance to
the actual performance.  The theoretical performance is obtained by modeling
the limiting factor of the execution time as either the critical path, or the sequential time
divided by the number of processors.
This is similar in approach to the Roofline
model~\cite{Williams:2009:RIV:1498765.1498785}.  Taking $\gamma_{seq}$ as the
sequential performance, $T$ as the total number of flops, $cp$ as the length of
the critical path, and $P$ as the number of processors, the predicted
performance, $\gamma_{pred}$, is
\[ \gamma_{pred} = \frac{\gamma_{seq} \cdot T}{\max \left( \frac{T}{P}, cp \right)} \]
Figures~\ref{fig.fig_tt_pic1_p40_z.th} and~\ref{fig.fig_tt_pic1_p40_d.th} depict
the predicted performance of all algorithms which use the \emph{Triangle on top
of triangle} kernels.  For double complex precision, sequential kernels reach
$3.1860$ GFlop/s while in double precision, the peak performance is $3.8440$
GFlop/s.  Since \PT provides an additional tuning parameter of the domain size,
we show the results for each value of this parameter as well as the composition
of the best of these domain sizes. Again, it is not evident what the domain size
should be for the best performance, hence our exhaustive search.

Part of our comprehensive study also involved comparisons made to the
Semi-Parallel Tile and Fully-Parallel Tile CAQR algorithms found
in~\cite{Hadri_enhancingparallelism} which are much the same as those found in
PLASMA.  As with PLASMA, the tuning parameter $\BS$ controls the domain size
upon which a flat tree is used to zero out tiles below the root tile within the
domain and a binary tree is used to merge these domains.  Unlike PLASMA, it is
not the bottom domain whose size decreases as the algorithm progresses through
the columns, but instead is the top domain.  In this study, we found that the
PLASMA algorithms performed identically or better than these algorithms and
therefore we do not report these comparisons.

\begin{sidewaysfigure*}
\centering
\subfloat[Predicted (double complex)\label{fig.fig_tt_pic1_p40_z.th}]{%
    \resizebox{.45\linewidth}{!}{\psfragfig{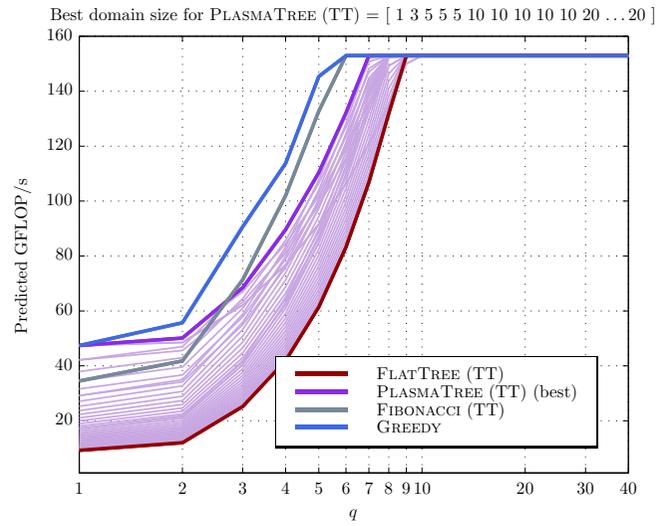}}%
}
\subfloat[Experimental (double complex)\label{fig.fig_tt_pic1_p40_z.exp}]{%
    \resizebox{.45\linewidth}{!}{\psfragfig{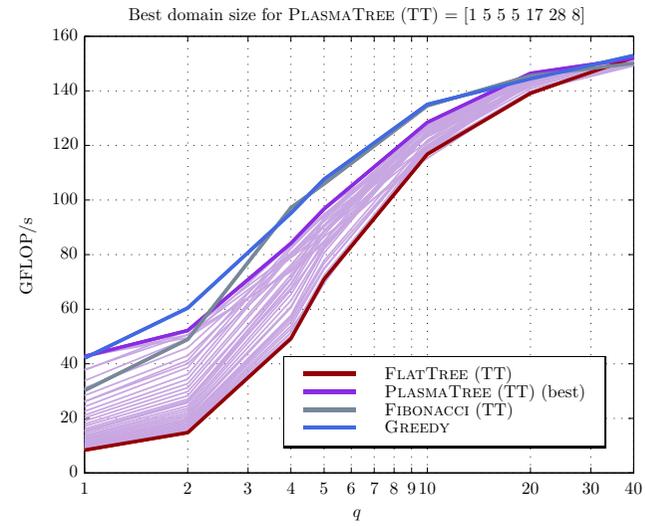}}%
}
\\
\subfloat[Predicted (double)\label{fig.fig_tt_pic1_p40_d.th}]{%
    \resizebox{.45\linewidth}{!}{\psfragfig{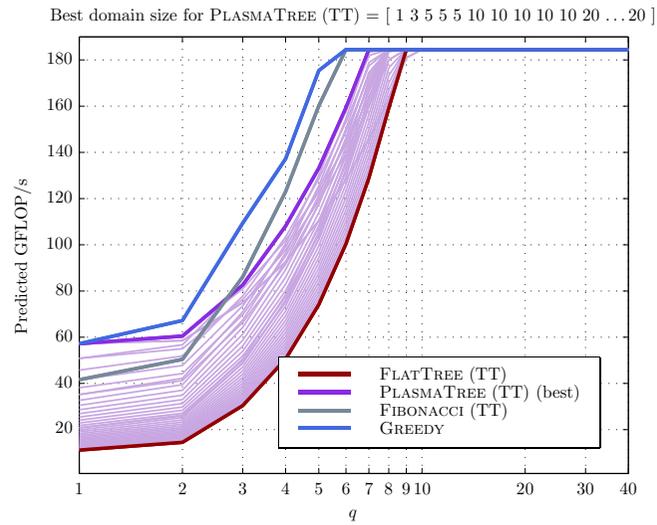}}%
}
\subfloat[Experimental (double)\label{fig.fig_tt_pic1_p40_d.exp}]{%
    \resizebox{.45\linewidth}{!}{\psfragfig{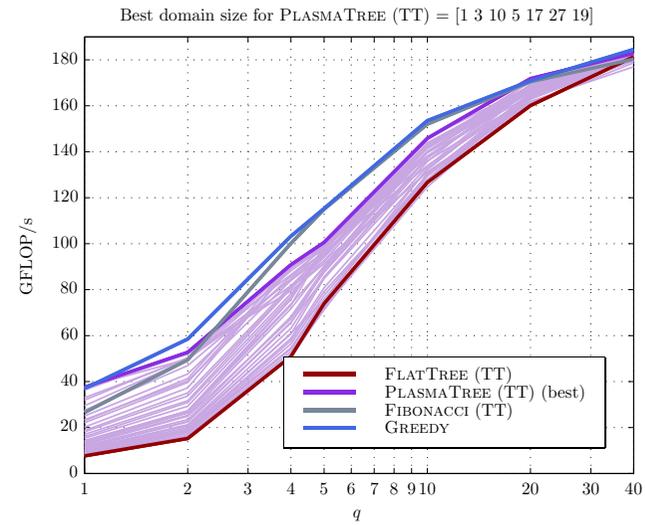}}%
}
\caption{\label{fig.fig_tt_pic1_p40}Predicted and experimental performance of QR factorization - \emph{Triangle on top of triangle} kernels}
\end{sidewaysfigure*}

\begin{sidewaysfigure*}
\centering
\begin{minipage}{\linewidth}
\subfloat[Theoretical CP length]{%
    \resizebox{.32\linewidth}{!}{\psfragfig{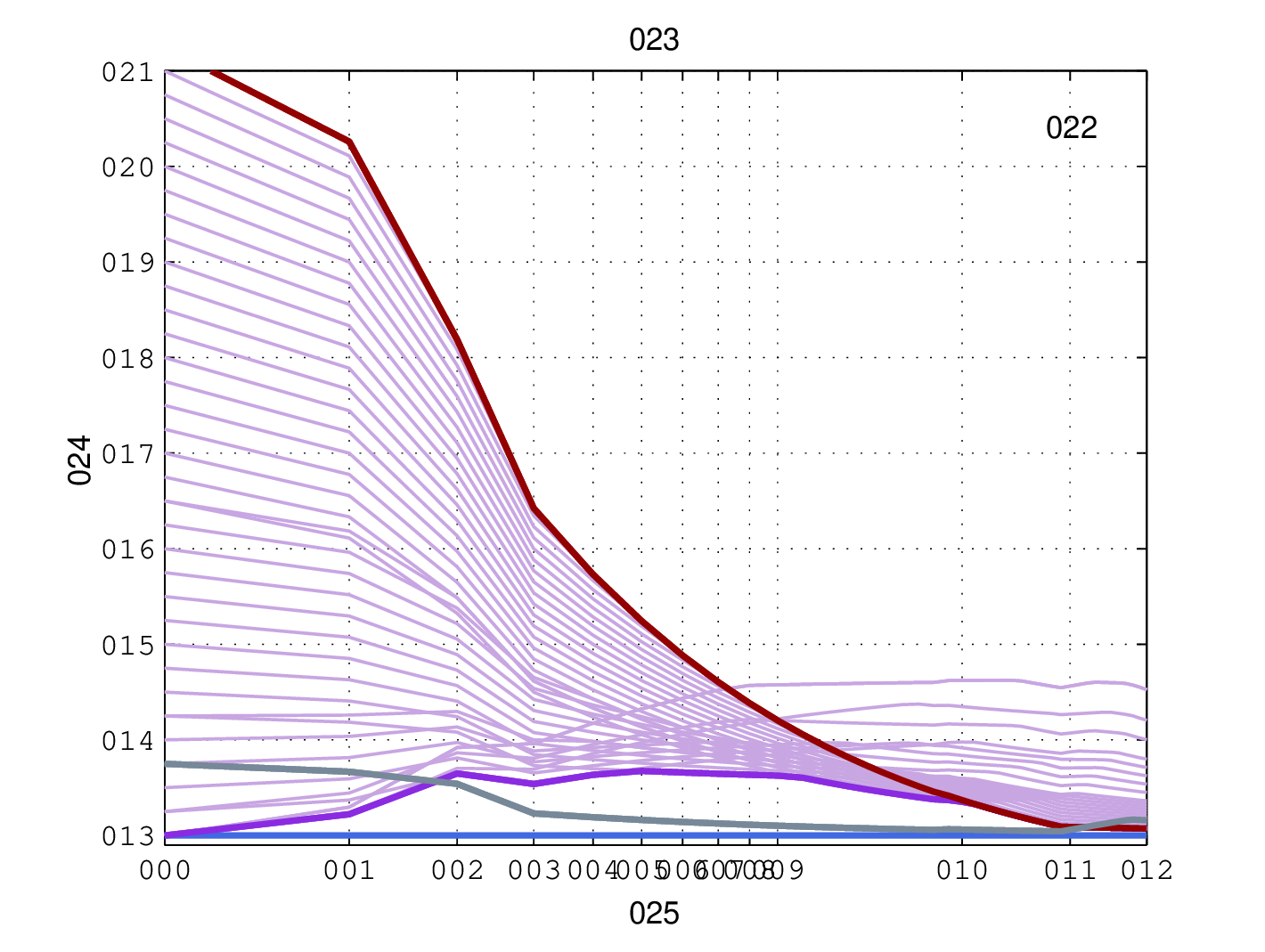}}%
}
\subfloat[Experimental (double complex)]{%
    \resizebox{.32\linewidth}{!}{\psfragfig{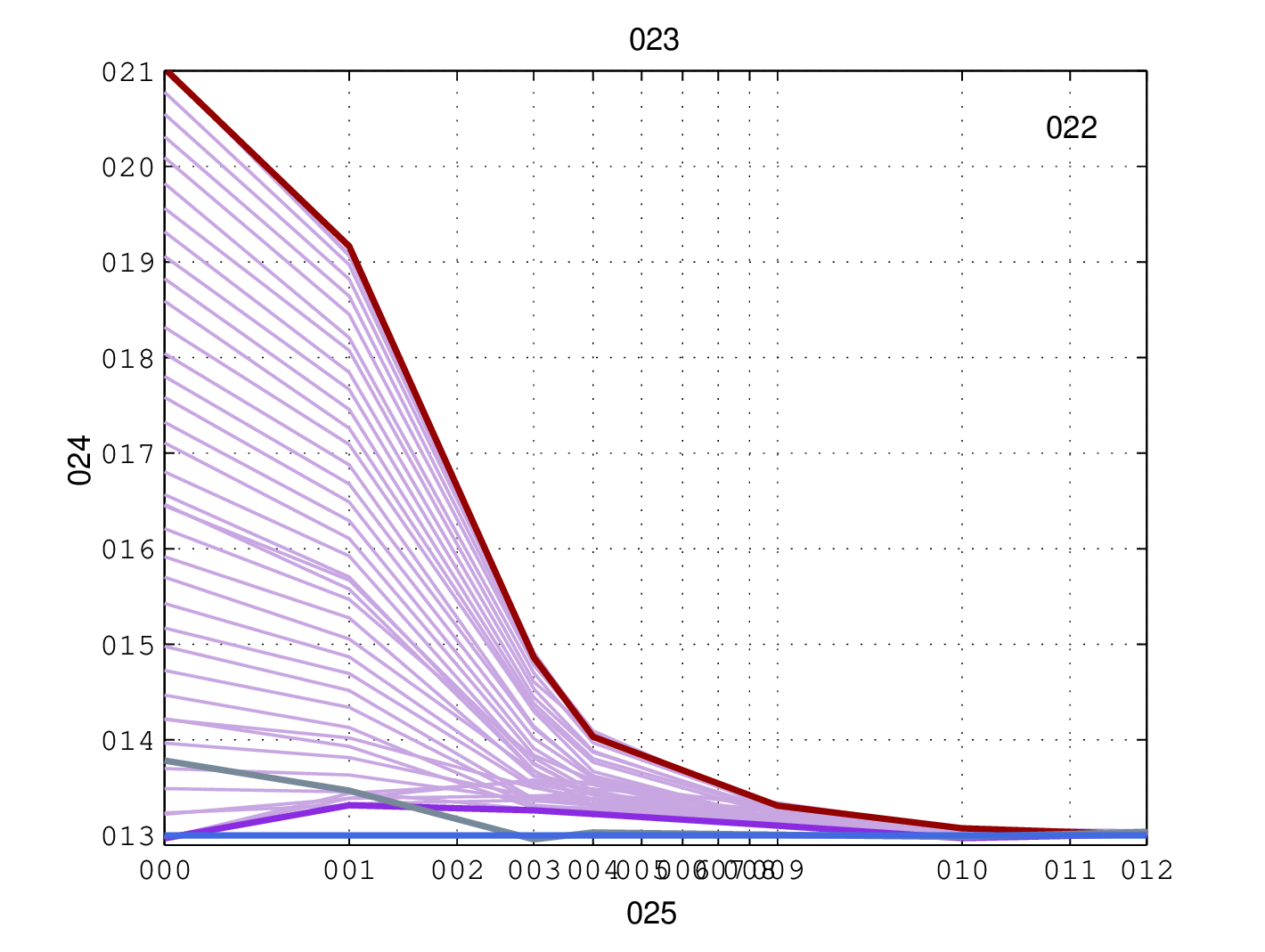}}%
}
\subfloat[Experimental (double)]{%
    \resizebox{.32\linewidth}{!}{\psfragfig{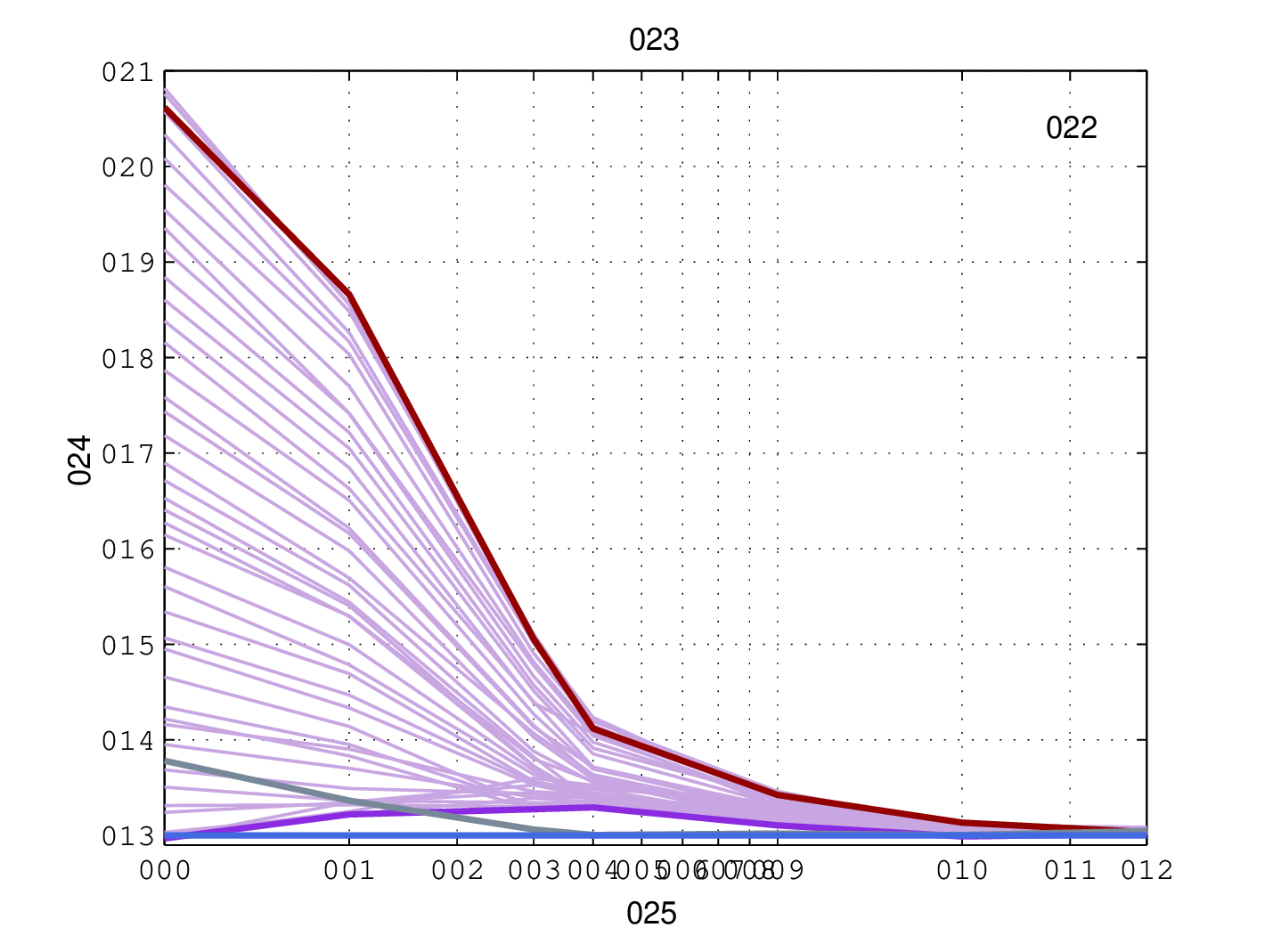}}%
}
\caption{\label{fig.fig_tt_pic2_p40} Overhead in terms of critical path length and time with respect to \Greedy (\Greedy = 1) }
\end{minipage}
\\
\begin{minipage}{\linewidth}
\subfloat[Theoretical CP length]{%
    \label{fig.fig_tt_pic3_p40.th}
    \resizebox{.32\linewidth}{!}{\psfragfig{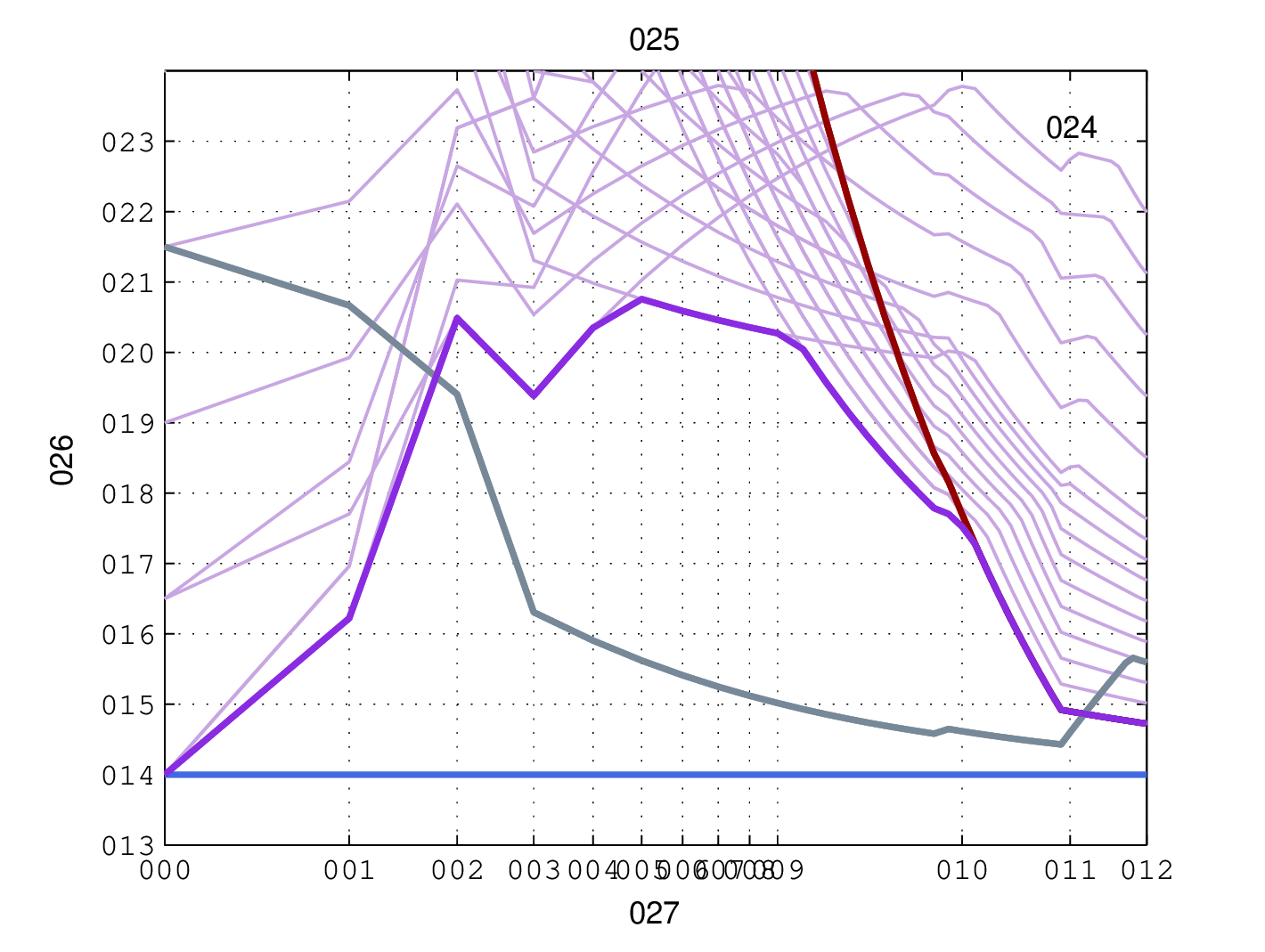}}%
}
\subfloat[Experimental (double complex)]{%
    \label{fig.fig_tt_pic3_p40_z.exp}
    \resizebox{.32\linewidth}{!}{\psfragfig{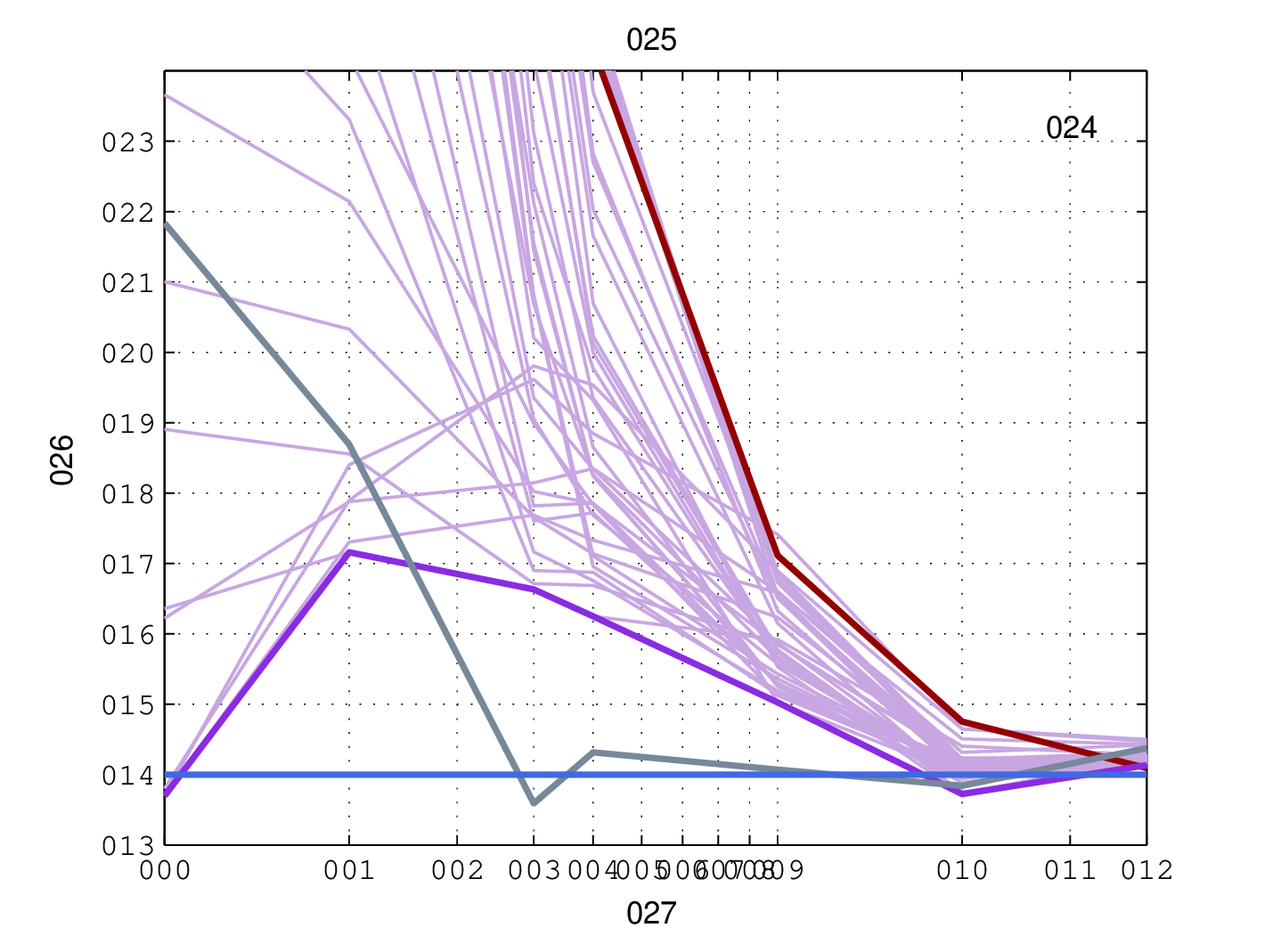}}%
}
\subfloat[Experimental (double)]{%
    \label{fig.fig_tt_pic3_p40_d.exp}
    \resizebox{.32\linewidth}{!}{\psfragfig{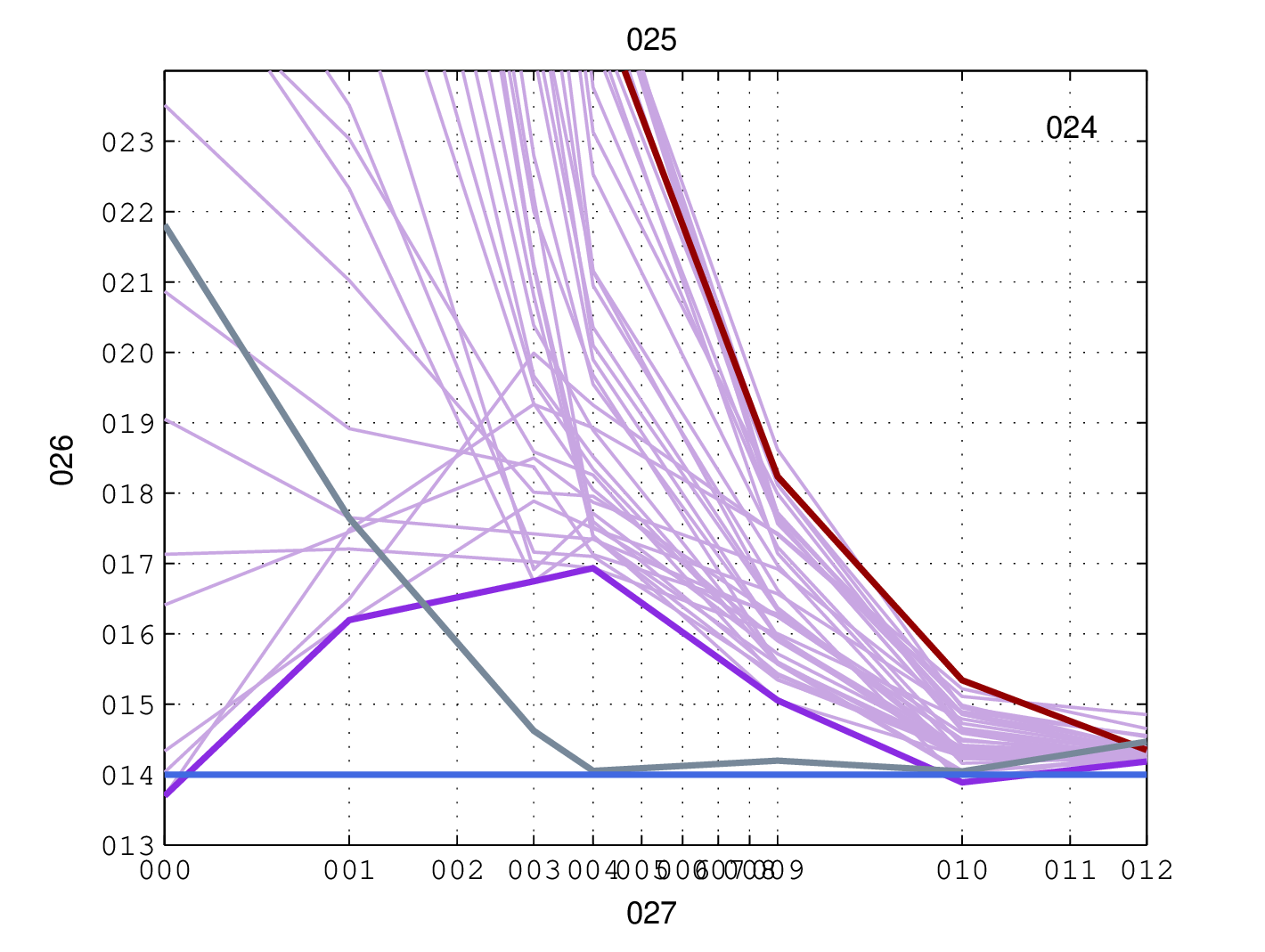}}%
}

\caption{\label{fig.fig_tt_pic3_p40} Detailed view of the overhead in terms of critical path length and time with respect to \Greedy (\Greedy = 1) }
\end{minipage}
\end{sidewaysfigure*}

Figure~\ref{fig.fig_tt_pic1_p40_z.exp} and~\ref{fig.fig_tt_pic1_p40_d.exp}
illustrate the experimental performance reached by \linebreak \Greedy, \MC and \PT
algorithms using the \emph{TT (Triangle on top of triangle)} kernels. In both cases,
double or double complex precision, the performance of \Greedy is better than
\PT even for the best choice of domain size.  Moreover, as expected from the
analysis in Section~\ref{sec:TiledAlgorithms}, \Greedy outperforms \MC the
majority of the time.  Furthermore, we see that, for rectangular matrices, the
experimental performance in double complex precision matches the prediction.
This is not the case for double precision because communications have higher
impact on performance.

While it is apparent that \Greedy does achieve higher levels of performance,
the percentage may not be as obvious.  To that end, taking \Greedy as the
baseline, we present in Figure~\ref{fig.fig_tt_pic2_p40} the theoretical,
double, and double complex precision overhead for each algorithm that uses the
\emph{Triangle on top of triangle} kernel as compared to \Greedy. These
overheads are respectively computed in terms of critical path length and time.
At a smaller scale (Figure~\ref{fig.fig_tt_pic3_p40}), it can be seen that
\Greedy can perform up to 13.6\% better than \PT.

For all matrix sizes considered, $p=40$ and $1 \leq q \leq 40$, in the
theoretical model, the critical path length for \Greedy is either the same as
that of \PT ($q=1$) or is up to 25\% shorter than \PT ($q=6$).  Analogously, the
critical path length for \Greedy is at least 2\% to 27\% shorter than that of
\MC.  In the experiments, the matrix sizes considered were $p=40$ and $q \in \{
1, 2, 4, 5, 10, 20, 40\}$.  In double precision, \Greedy has a decrease of at
most 1.5\% than the best \PT ($q=1$) and a gain of at most 12.8\% than the best
\PT ($q=5$).  In double complex precision, \Greedy has a decrease of at most
1.5\% than the best \PT ($q=1$) and a gain of at most 13.6\% than the best \PT
($q=2$).  Similarly, in double precision, \Greedy provides a gain of 2.6\% to
28.1\% over \MC and in double complex precision, \Greedy has a decrease of at
most 2.1\% and a gain of at most 28.2\% over \MC.

Although it is evidenced that \PT does not vary too far from \Greedy or \MC,
one must keep in mind that there is a tuning parameter involved and we choose
the best of these domain sizes for \PT to create the composite result, whereas
with \Greedy, there is no such parameter to consider.  Of particular interest
is the fact that \Greedy always performs better than any other algorithm\footnote{When $q=1$, \Greedy and \FT exhibit close performance. They
both perform a binary tree reduction, albeit with different row pairings.} for $
p \gg q$. In
the scope of \PT, a domain size $\BS=1$ will force the use of a binary tree so
that both \Greedy and \PT behave the same. However, as the matrix tends more to
a square, i.e., $q$ tends toward $p$, we observe that the performance of all of
the algorithms, including \FT, are on par with \Greedy.  As more columns are
added, the parallelism of the algorithm is increased and the critical path
becomes less of a limiting factor, so that the performance of the kernels is
brought to the forefront.  Therefore, all of the algorithms are performing
similarly since they all share the same kernels.

\begin{sidewaysfigure*}
\centering
\begin{minipage}{0.48\linewidth}
\subfloat[Factorization kernels]{
    \resizebox{\linewidth}{!}{\psfragfig{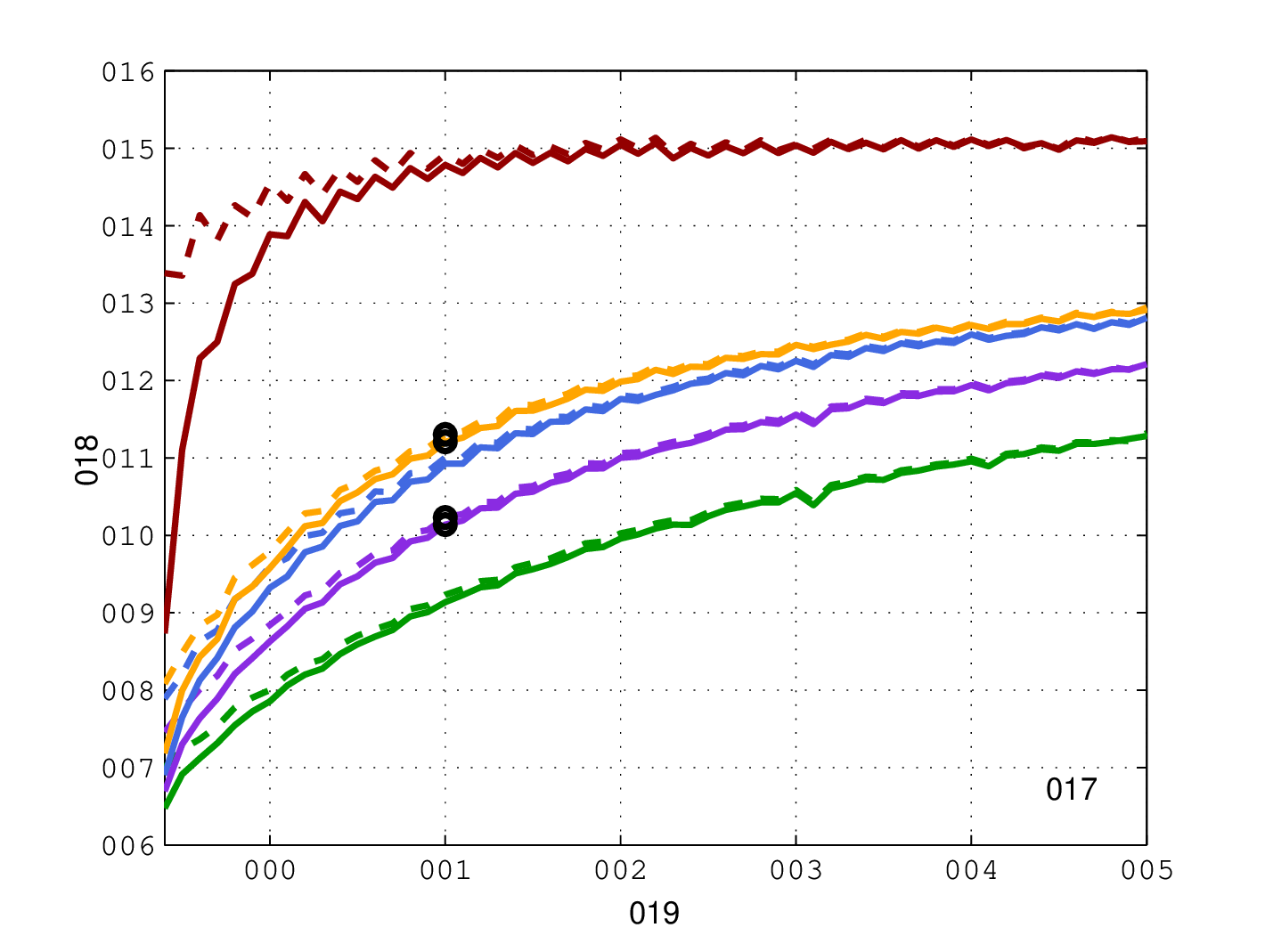}}%
}
\\
\subfloat[Update kernels]{
    \resizebox{\linewidth}{!}{\psfragfig{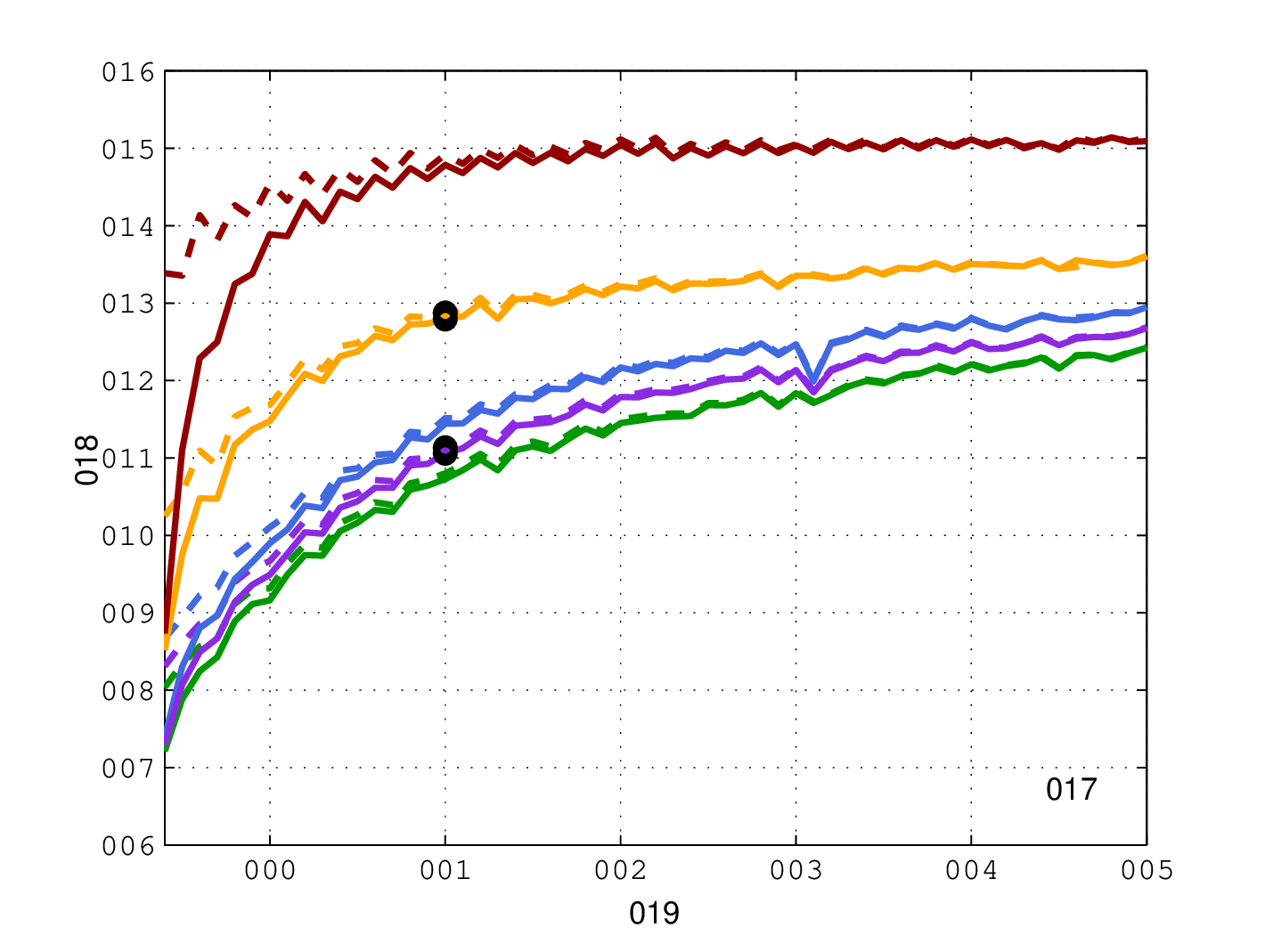}}%
}
    \caption{\label{fig.perf_kernels_z} Kernel performance for double complex precision}
\end{minipage}
~~
\begin{minipage}{0.48\linewidth}
\subfloat[Factorization kernels]{
    \resizebox{\linewidth}{!}{\psfragfig{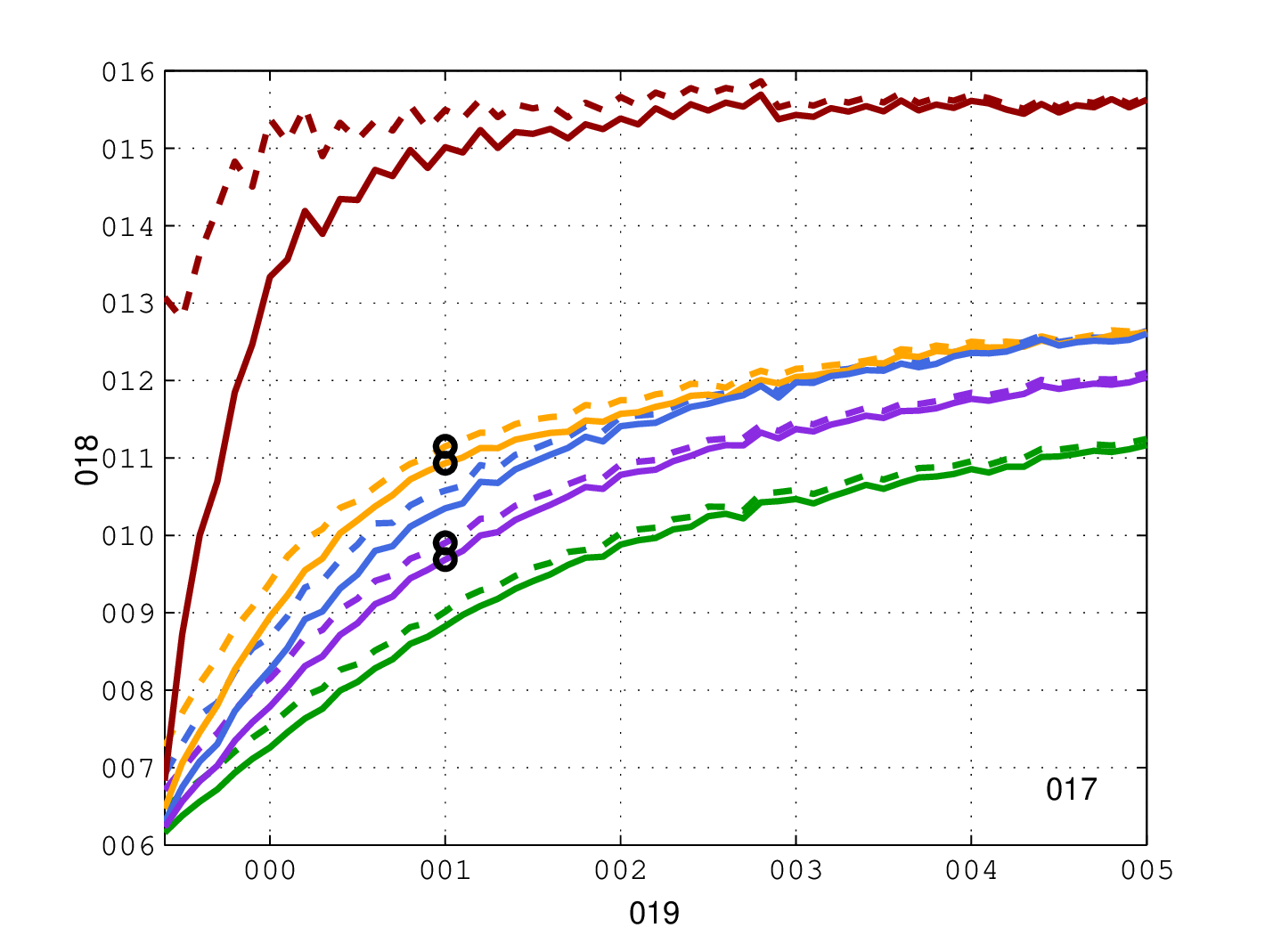}}%
}
\\
\subfloat[Update kernels]{
    \resizebox{\linewidth}{!}{\psfragfig{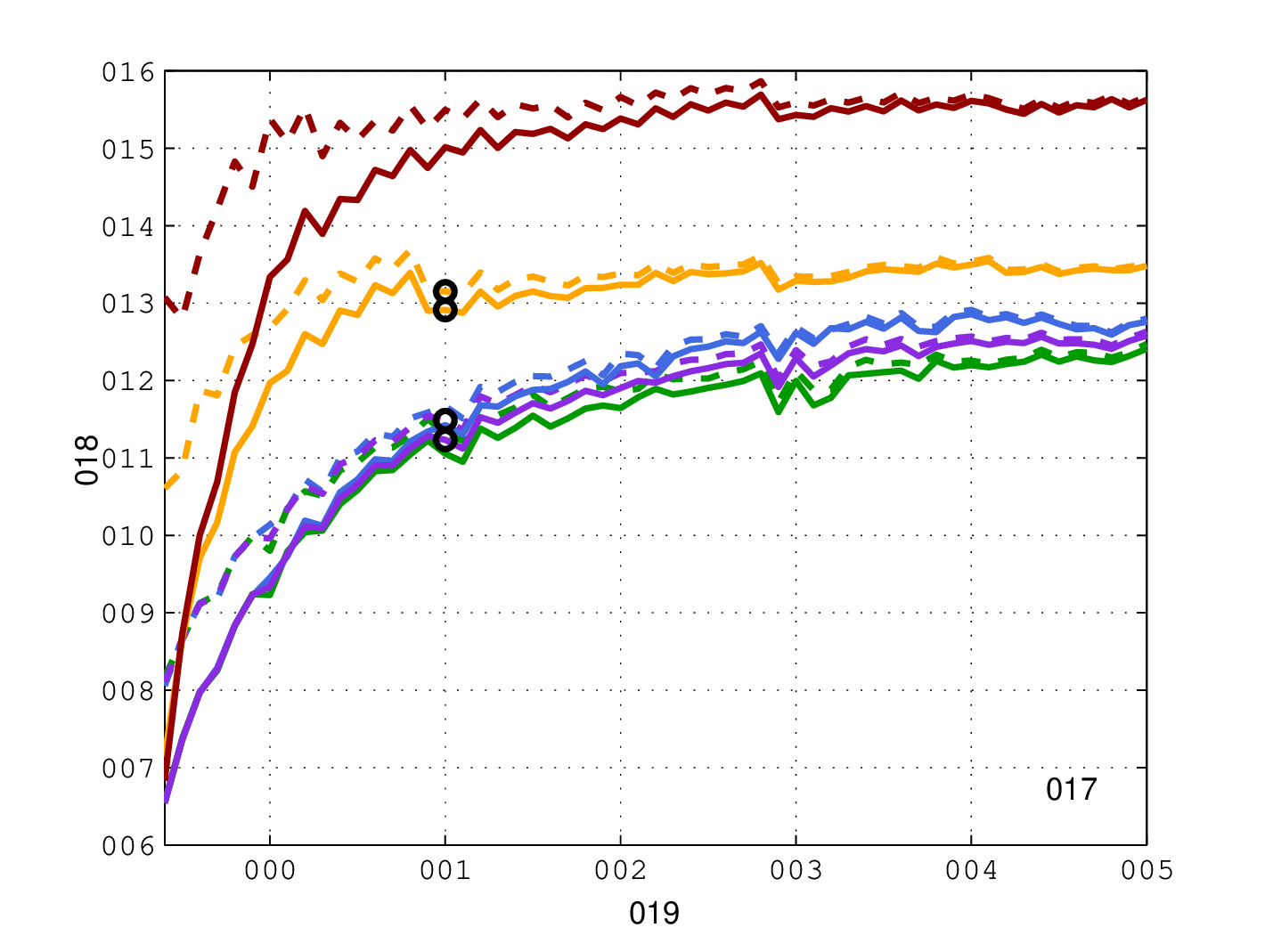}}%
}
    \caption{\label{fig.perf_kernels_d} Kernel performance for double precision}
\end{minipage}
\end{sidewaysfigure*}

In order to accurately assess the impact of the kernel selection towards the
performance of the algorithms, Figures~\ref{fig.perf_kernels_z}
and~\ref{fig.perf_kernels_d} show both the in cache and out of cache
performance using the \emph{No Flush} and \emph{MultCallFlushLRU} strategies as
presented in~\cite{lawn242,Whaley:2008:AAC:1462062.1462065}.  Since an
algorithm using \emph{TT} kernels will need to call
\GEQRT as well as \TTQRT to achieve the same as the \emph{TS} kernel \TSQRT, the comparison is made
between \GEQRT + \TTQRT and \TSQRT (and similarly for the updates).  For $n_b=200$, the observed ratio for in
cache kernel speed for \TSQRT to \GEQRT + \TTQRT is 1.3374, and for \TSMQR to
\UNMQR + \TTMQR is 1.3207. For out of cache, the ratio for \TSQRT to \GEQRT
+ \TTQRT is 1.3193 and for \TSMQR to \UNMQR + \TTMQR it is 1.3032.  Thus, we
can expect about a 30\% difference between the selection of the kernels, since
we will have instances of using in cache and out of cache throughout the run.
Most of this difference is due to the higher efficiency and data locality
within the \emph{TT} kernels as compared to the
\emph{TS} kernels.

\begin{sidewaysfigure*}
\centering
\subfloat[Predicted (double complex)]{%
    \label{fig.fig_aa_pic1_p40_z.th}%
    \resizebox{.45\linewidth}{!}{\psfragfig{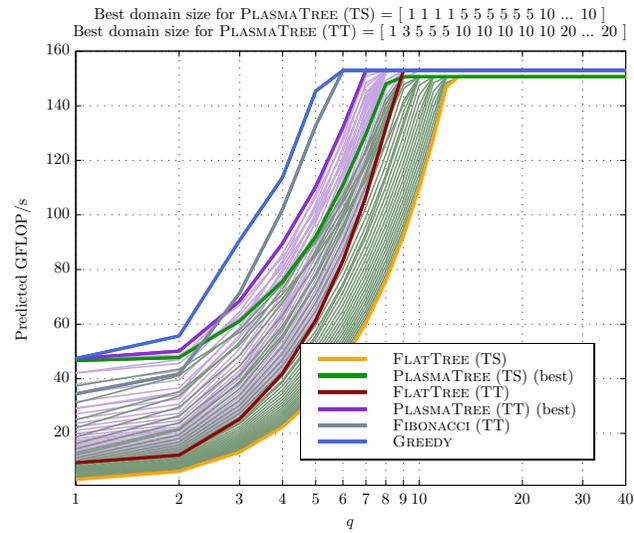}}%
}
\subfloat[Experimental (double complex)]{%
    \label{fig.fig_aa_pic1_p40_z.exp}%
    \resizebox{.45\linewidth}{!}{\psfragfig{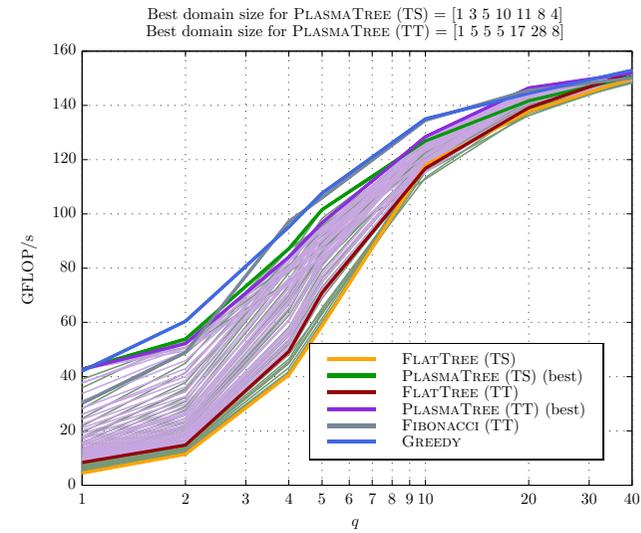}}%
}
\\
\subfloat[Predicted (double)]{%
    \label{fig.fig_aa_pic1_p40_d.th}%
    \resizebox{.45\linewidth}{!}{\psfragfig{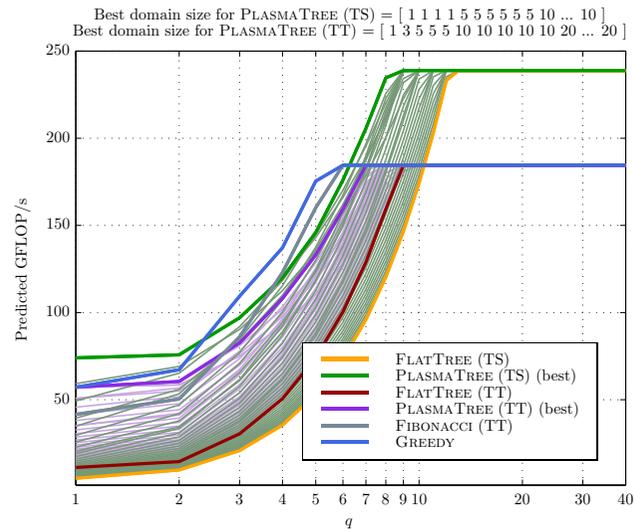}}%
}
\subfloat[Experimental (double)]{%
    \label{fig.fig_aa_pic1_p40_d.exp}%
    \resizebox{.45\linewidth}{!}{\psfragfig{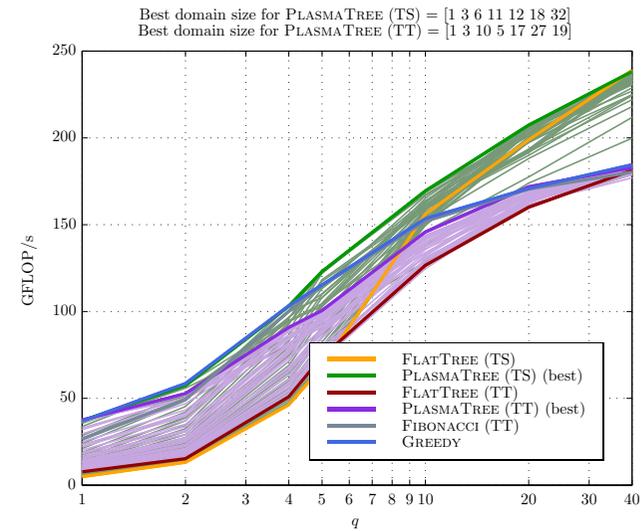}}%
}
\caption{\label{fig.fig_aa_pic1_p40}
Predicted and experimental performance of QR factorization - All kernels}
\end{sidewaysfigure*}

\begin{sidewaysfigure*}
\centering
\begin{minipage}{\linewidth}
\subfloat[Theoretical CP length]{%
    \resizebox{.32\linewidth}{!}{\psfragfig{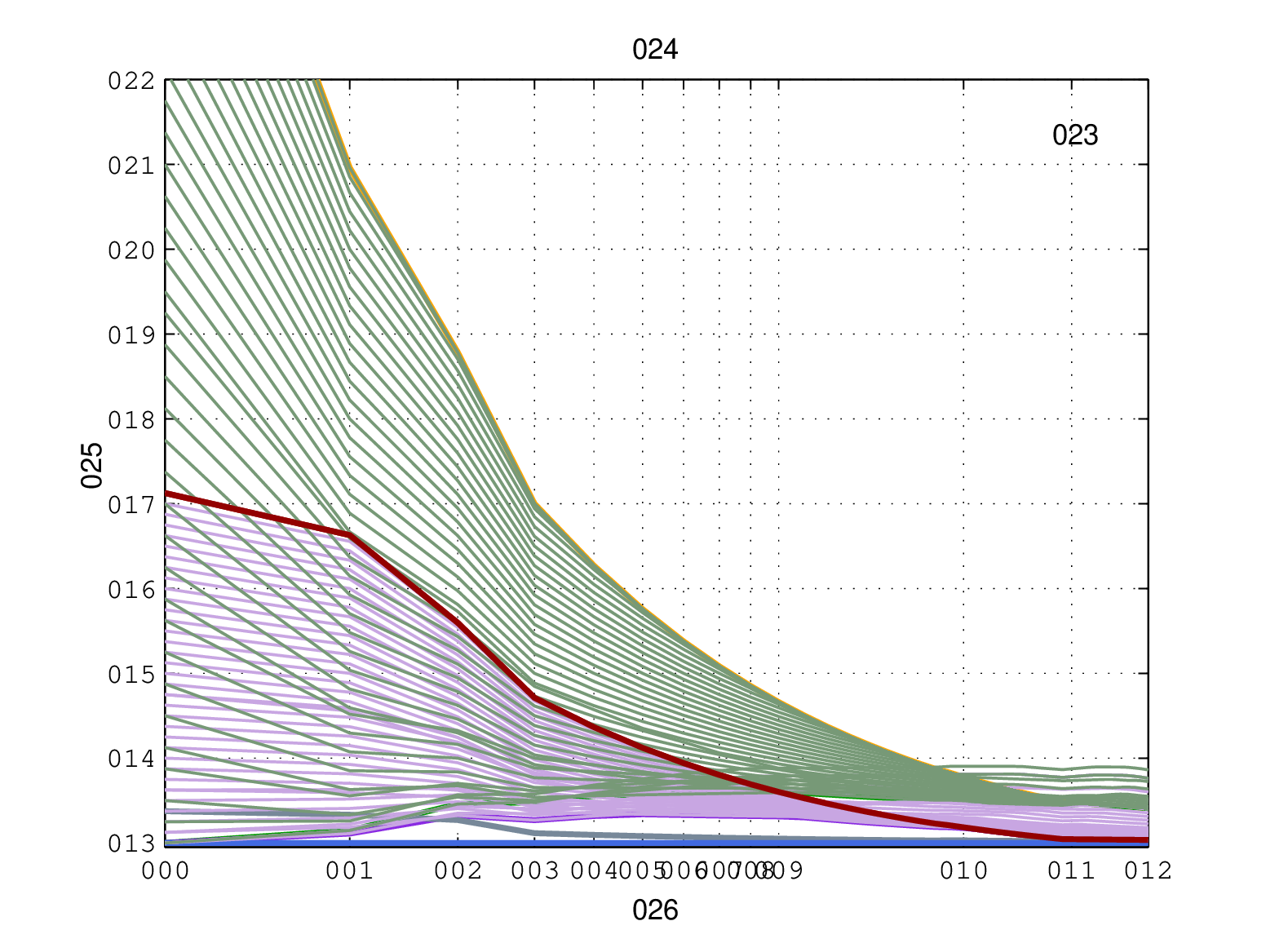}}%
}
\subfloat[Experimental (double complex)]{%
    \resizebox{.32\linewidth}{!}{\psfragfig{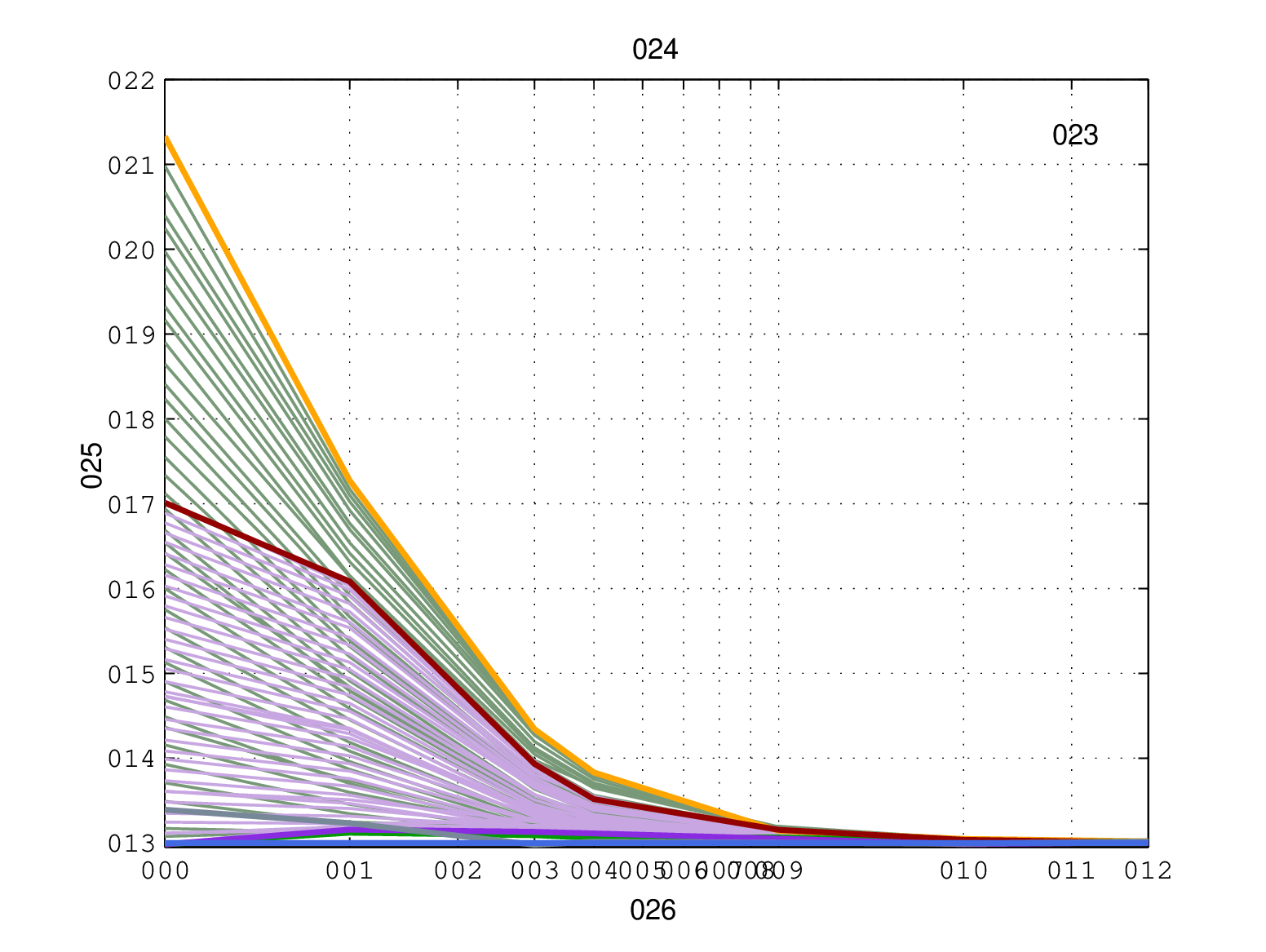}}%
}
\subfloat[Experimental (double)]{%
    \resizebox{.32\linewidth}{!}{\psfragfig{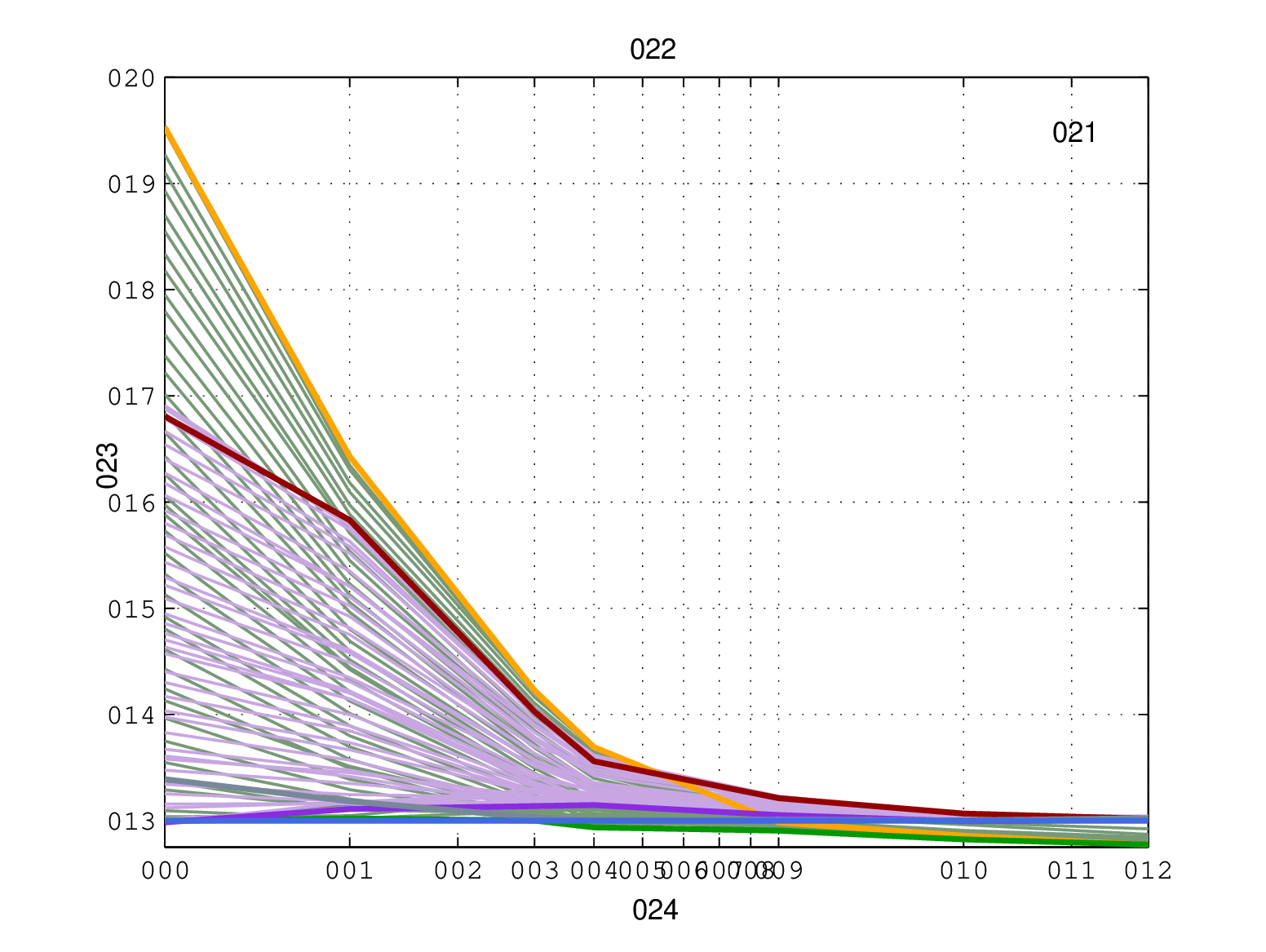}}%
}
\caption{\label{fig.fig_aa_pic2_p40} Overhead in terms of critical path length and time with respect to \Greedy (\Greedy = 1) }
\end{minipage}
\\
\begin{minipage}{\linewidth}
\subfloat[Theoretical CP length]{%
    \label{fig.fig_aa_pic3_p40.th}
    \resizebox{.32\linewidth}{!}{\psfragfig{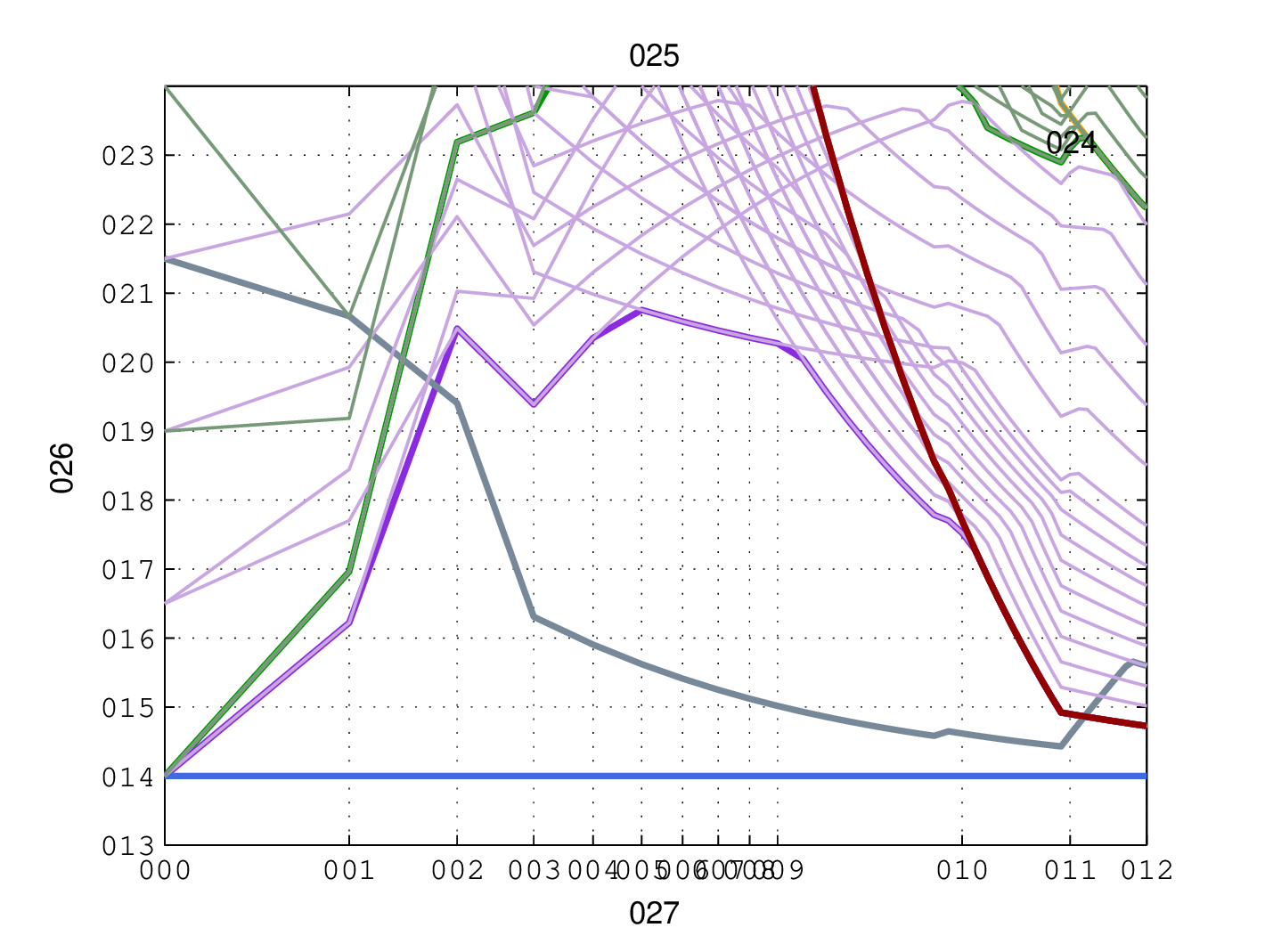}}%
}
\subfloat[Experimental (double complex)]{%
    \label{fig.fig_aa_pic3_p40_z.exp}
    \resizebox{.32\linewidth}{!}{\psfragfig{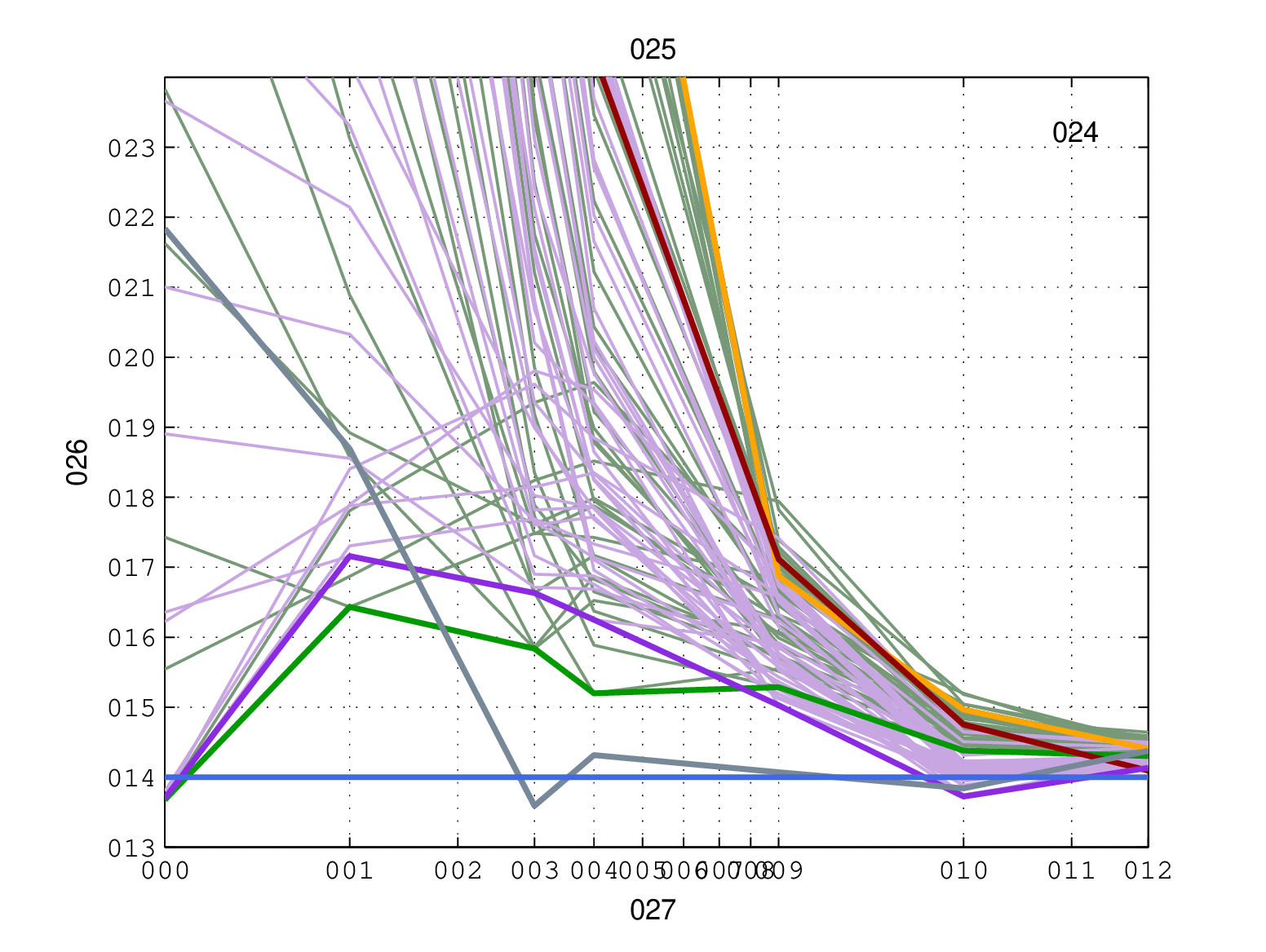}}%
}
\subfloat[Experimental (double)]{%
    \label{fig.fig_aa_pic3_p40_d.exp}
    \resizebox{.32\linewidth}{!}{\psfragfig{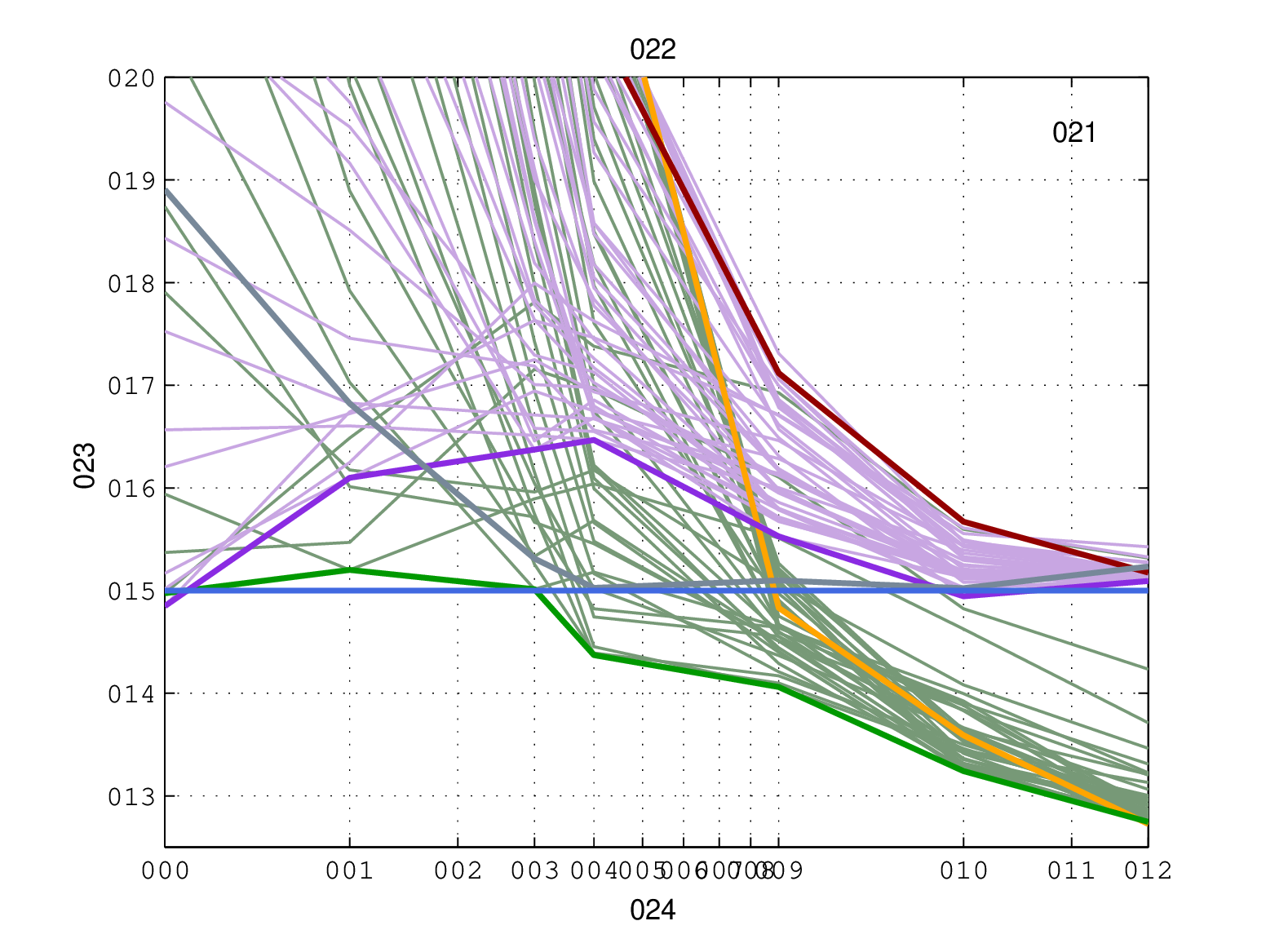}}%
}

\caption{\label{fig.fig_aa_pic3_p40} Detailed view of the overhead in terms of critical path length and time with respect to \Greedy (\Greedy = 1) }
\end{minipage}
\end{sidewaysfigure*}

Having seen that kernel performance can have a significant impact,
we also compare the \emph{TT} based
algorithms to those using the \emph{TS} kernels. The goal is to provide
a complete assessment of all currently available
algorithms,  as shown
in Figure~\ref{fig.fig_aa_pic1_p40}.  For double precision, the observed
difference in kernel speed is 4.976 GFLOP/sec for the \emph{TS} kernels
versus 3.844 GFLOP/sec for the \emph{TT} kernels which provides a ratio of 1.2945 and is in accordance with
our previous analysis.  It can be seen that as the number of columns increases,
whereby the amount of parallelism increases, the effect of the kernel
performance outweighs the benefit provided by the extra parallelism afforded
through the \emph{TT} algorithms.  Comparatively, in
double complex precision, \Greedy does perform better, even against the
algorithms using the \emph{TS} kernels.  As before, one
must keep in mind that \Greedy does not require the tuning parameter of the
domain size to achieve this better performance.

From these experiments, we showed that in double complex precision, \linebreak \Greedy
demonstrated better performance than any of the other algorithms and moreover,
it does so without the need to specify a domain size as opposed to the
algorithms in PLASMA. In addition, in double precision, for matrices where $p
\gg q$, \Greedy continues to excel over any other algorithm using
the \emph{TT} kernels, and continues to do
so as the matrices become more square.

\begin{table*}
    \centering
    \resizebox{\linewidth}{!}{
    \begin{tabular}{|r|r|r||r|r|r|r||r|r|r|}
       \hline
        $p$ & $q$ & \Greedy & \PT (TT) &    \BS & Overhead & Gain    &  \MC & Overhead &    Gain\\
       \hline
       40 & 1 &     16 &     16 &     1 &   1.0000 &  0.0000 &   22 &   1.3750 &  0.2727\\
       40 & 2 &     54 &     60 &     3 &   1.1111 &  0.1000 &   72 &   1.3333 &  0.2500\\
       40 & 3 &     74 &     98 &     5 &   1.3243 &  0.2449 &   94 &   1.2703 &  0.2128\\
       40 & 4 &    104 &    132 &     5 &   1.2692 &  0.2121 &  116 &   1.1154 &  0.1034\\
       40 & 5 &    126 &    166 &     5 &   1.3175 &  0.2410 &  138 &   1.0952 &  0.0870\\
       40 & 6 &    148 &    198 &    10 &   1.3378 &  0.2525 &  160 &   1.0811 &  0.0750\\
       40 & 7 &    170 &    226 &    10 &   1.3294 &  0.2478 &  182 &   1.0706 &  0.0659\\
       40 & 8 &    192 &    254 &    10 &   1.3229 &  0.2441 &  204 &   1.0625 &  0.0588\\
       40 & 9 &    214 &    282 &    10 &   1.3178 &  0.2411 &  226 &   1.0561 &  0.0531\\
       40 &10 &    236 &    310 &    10 &   1.3136 &  0.2387 &  248 &   1.0508 &  0.0484\\
       40 &11 &    258 &    336 &    20 &   1.3023 &  0.2321 &  270 &   1.0465 &  0.0444\\
       40 &12 &    280 &    358 &    20 &   1.2786 &  0.2179 &  292 &   1.0429 &  0.0411\\
       40 &13 &    302 &    380 &    20 &   1.2583 &  0.2053 &  314 &   1.0397 &  0.0382\\
       40 &14 &    324 &    402 &    20 &   1.2407 &  0.1940 &  336 &   1.0370 &  0.0357\\
       40 &15 &    346 &    424 &    20 &   1.2254 &  0.1840 &  358 &   1.0347 &  0.0335\\
       40 &16 &    368 &    446 &    20 &   1.2120 &  0.1749 &  380 &   1.0326 &  0.0316\\
       40 &17 &    390 &    468 &    20 &   1.2000 &  0.1667 &  402 &   1.0308 &  0.0299\\
       40 &18 &    412 &    490 &    20 &   1.1893 &  0.1592 &  424 &   1.0291 &  0.0283\\
       40 &19 &    432 &    512 &    20 &   1.1852 &  0.1562 &  446 &   1.0324 &  0.0314\\
       40 &20 &    454 &    534 &    20 &   1.1762 &  0.1498 &  468 &   1.0308 &  0.0299\\
       40 &21 &    476 &    554 &    20 &   1.1639 &  0.1408 &  490 &   1.0294 &  0.0286\\
       40 &22 &    498 &    570 &    20 &   1.1446 &  0.1263 &  512 &   1.0281 &  0.0273\\
       40 &23 &    520 &    586 &    20 &   1.1269 &  0.1126 &  534 &   1.0269 &  0.0262\\
       40 &24 &    542 &    602 &    20 &   1.1107 &  0.0997 &  556 &   1.0258 &  0.0252\\
       40 &25 &    564 &    618 &    20 &   1.0957 &  0.0874 &  578 &   1.0248 &  0.0242\\
       40 &26 &    586 &    634 &    20 &   1.0819 &  0.0757 &  600 &   1.0239 &  0.0233\\
       40 &27 &    608 &    650 &    20 &   1.0691 &  0.0646 &  622 &   1.0230 &  0.0225\\
       40 &28 &    630 &    666 &    20 &   1.0571 &  0.0541 &  644 &   1.0222 &  0.0217\\
       40 &29 &    652 &    682 &    20 &   1.0460 &  0.0440 &  666 &   1.0215 &  0.0210\\
       40 &30 &    668 &    698 &    20 &   1.0449 &  0.0430 &  688 &   1.0299 &  0.0291\\
       40 &31 &    684 &    714 &    20 &   1.0439 &  0.0420 &  710 &   1.0380 &  0.0366\\
       40 &32 &    700 &    730 &    20 &   1.0429 &  0.0411 &  732 &   1.0457 &  0.0437\\
       40 &33 &    716 &    746 &    20 &   1.0419 &  0.0402 &  754 &   1.0531 &  0.0504\\
       40 &34 &    732 &    762 &    20 &   1.0410 &  0.0394 &  776 &   1.0601 &  0.0567\\
       40 &35 &    748 &    778 &    20 &   1.0401 &  0.0386 &  798 &   1.0668 &  0.0627\\
       40 &36 &    764 &    794 &    20 &   1.0393 &  0.0378 &  820 &   1.0733 &  0.0683\\
       40 &37 &    780 &    810 &    20 &   1.0385 &  0.0370 &  842 &   1.0795 &  0.0736\\
       40 &38 &    796 &    826 &    20 &   1.0377 &  0.0363 &  862 &   1.0829 &  0.0766\\
       40 &39 &    812 &    842 &    20 &   1.0369 &  0.0356 &  878 &   1.0813 &  0.0752\\
       40 &40 &    826 &    856 &    20 &   1.0363 &  0.0350 &  892 &   1.0799 &  0.0740\\
       \hline
    \end{tabular}
    }
    \caption{\Greedy versus \PT (TT) and \MC (Theoretical)}
\end{table*}
\begin{table*}
    \centering
    \centering
    \begin{tabular}{|r|r|r||r|r|r|r|}
       \hline
         $p$ &  $q$ & \Greedy & \PT (TT) & \BS & Overhead &   Gain\\
         \hline
        40 &  1 &   36.9360 &   37.5020 &  1 &   1.0153 &-0.0153\\
        40 &  2 &   58.5090 &   52.7180 &  3 &   0.9010 & 0.0990\\
        40 &  4 &  103.2670 &   90.7940 & 10 &   0.8792 & 0.1208\\
        40 &  5 &  115.3060 &  100.5540 &  5 &   0.8721 & 0.1279\\
        40 & 10 &  153.5180 &  145.8200 & 17 &   0.9499 & 0.0501\\
        40 & 20 &  170.8730 &  171.8270 & 27 &   1.0056 &-0.0056\\
        40 & 40 &  184.5220 &  182.8160 & 19 &   0.9908 & 0.0092\\
       \hline
    \end{tabular}
    \caption{\Greedy versus \PT (TT) (Experimental Double)}
\end{table*}
\begin{table*}
    \centering
    \begin{tabular}{|r|r|r||r|r|r|r|}
       \hline
         $p$ &  $q$ & \Greedy & \PT (TT) & \BS & Overhead &   Gain\\
         \hline
        40 &  1 &   42.0710 &   42.7120 &  1 &   1.0152 &-0.0152\\
        40 &  2 &   60.4420 &   52.1970 &  5 &   0.8636 & 0.1364\\
        40 &  4 &   95.1820 &   84.1120 &  5 &   0.8837 & 0.1163\\
        40 &  5 &  107.6370 &   96.7530 &  5 &   0.8989 & 0.1011\\
        40 & 10 &  135.0270 &  128.4320 & 17 &   0.9512 & 0.0488\\
        40 & 20 &  144.4010 &  146.4220 & 28 &   1.0140 &-0.0140\\
        40 & 40 &  152.9280 &  151.9090 &  8 &   0.9933 & 0.0067\\
       \hline
    \end{tabular}
    \caption{\Greedy versus \PT (TT) (Experimental Double Complex)}
\end{table*}
\begin{table*}
    \centering
    \begin{tabular}{|r|r|r||r|r|r|}
         \hline
         $p$ &  $q$ & \Greedy &    \MC & Overhead &   Gain\\
         \hline
        40 &  1 &   36.9360 &  26.5610  &   0.7191 & 0.2809\\
        40 &  2 &   58.5090 &  49.4870  &   0.8458 & 0.1542\\
        40 &  4 &  103.2670 & 100.1440  &   0.9698 & 0.0302\\
        40 &  5 &  115.3060 & 115.0020  &   0.9974 & 0.0026\\
        40 & 10 &  153.5180 & 152.0090  &   0.9902 & 0.0098\\
        40 & 20 &  170.8730 & 170.4780  &   0.9977 & 0.0023\\
        40 & 40 &  184.5220 & 180.2990  &   0.9771 & 0.0229\\
         \hline
    \end{tabular}
    \caption{Greedy versus \MC (Experimental Double)}
\end{table*}
\begin{table*}
    \centering
    \begin{tabular}{|r|r|r||r|r|r|}
         \hline
         $p$ &  $q$ & \Greedy &    \MC & Overhead &   Gain\\
         \hline
        40 &  1 &   42.0710 &  30.2280  &   0.7185 & 0.2815\\
        40 &  2 &   60.4420 &  48.9570  &   0.8100 & 0.1900\\
        40 &  4 &   95.1820 &  97.1650  &   1.0208 &-0.0208\\
        40 &  5 &  107.6370 & 105.9610  &   0.9844 & 0.0156\\
        40 & 10 &  135.0270 & 134.5500  &   0.9965 & 0.0035\\
        40 & 20 &  144.4010 & 145.5530  &   1.0080 &-0.0080\\
        40 & 40 &  152.9280 & 150.0980  &   0.9815 & 0.0185\\
         \hline
    \end{tabular}
    \caption{\Greedy versus \MC (Experimental Double Complex)}
\end{table*}

\section{Conclusion}
\label{sec.conclusion}

In this manuscript, we have presented \MC, and \Greedy, two new algorithms for tiled QR factorization.
These algorithms exhibit more parallelism than state-of-the-art implementations
based on reduction trees. We have provided accurate estimations for the length of their critical path,
and we have proven that they were asymptotically optimal for a wide class of matrix shapes, including
all cases where the number of tile rows $p$ and tile columns $q$ are proportional, $p = \lambda q$, $\lambda \geq 1$.
To the best of our knowledge, this proof is the first complexity result in the field of tiled algorithms,
and it lays the theoretical foundations for a comparative study of these algorithms.

Comprehensive experiments on multicore platforms confirm the superiority of the new algorithms for $p \times q$ matrices,
as soon as, say, $p \geq 2q$.
This holds true when comparing not only with previous algorithms using TT (\emph{Triangle on top of triangle}) kernels, but also
with all known algorithms based on TS (\emph{Triangle on top of square}) kernels. Given that TS kernels offer more locality, and
benefit from better elementary arithmetic performance, than TT kernels, the better performance of the new algorithms is
even more striking, and further demonstrates that a large degree of a parallelism was not exploited in previously published solutions.

Future work will investigate several promising directions. First, using rectangular tiles instead of square tiles
could lead to efficient algorithms, with more locality and still the same potential for parallelism. Second,
refining the model to account for communications, and extending it to fully distributed architectures, would lay the ground to the design of
MPI implementations of the new algorithms, unleashing their high level of performance on larger platforms.
Finally, the design of robust algorithms, capable of achieving efficient performance despite variations in processor speeds,
or even resource failures, is a challenging but crucial task to fully benefit from future platforms with a huge number of cores.

\bibliographystyle{abbrv}
\bibliography{biblio}

\end{document}